\documentclass[lettersize,journal]{IEEEtran}

\usepackage{algorithmic}
\usepackage[linesnumbered, ruled]{algorithm2e}
\usepackage{amsmath,amsfonts}
\usepackage{amsthm}
\usepackage{amssymb}
\usepackage{calligra}
\usepackage{algorithmic}
\usepackage{array}
\usepackage{bbm}
\usepackage[caption=false,font=normalsize,labelfont=sf,textfont=sf]{subfig}
\usepackage{textcomp}
\usepackage{stfloats}
\usepackage{url}
\usepackage{verbatim}
\usepackage{graphicx}
\def\BibTeX{{\rm B\kern-.05em{\sc i\kern-.025em b}\kern-.08em
    T\kern-.1667em\lower.7ex\hbox{E}\kern-.125emX}}
\usepackage{balance}
\usepackage{makecell}
\usepackage{multirow}
\usepackage{lipsum}
\usepackage{mathrsfs}
\usepackage[numbers,sort&compress]{natbib}

\newtheorem{theorem}{Theorem}
\newtheorem{lemma}{Lemma}
\newtheorem{corollary}{Corollary}
\newtheorem{assumption}{Assumption}

\begin{document}
\title{Infinite Factorial Linear Dynamical Systems for Transient Signal Detection}
\author{Jiadi Bao,~\IEEEmembership{Student Member,~IEEE,}
Yatong Wang,~\IEEEmembership{Member,~IEEE,}
Yunjie Li,~\IEEEmembership{Senior Member,~IEEE,}\\
Mengtao Zhu,~\IEEEmembership{Member,~IEEE,}
and Shafei Wang
\thanks{Jiadi Bao, Yatong Wang, Yunjie Li, Mengtao Zhu are with the School of Information and Electronics, Beijing Institute of Technology, Beijing, 100081, China (E-mail: baojiadi@bit.edu.cn; wangyatong@bit.edu.cn, liyunjie@bit.edu.cn; zhumengtao@bit.edu.cn;).}
\thanks{Shafei Wang is with the School of Information and Electronics, Beijing Institute of Technology, Beijing, 100081, China, and also with the Laboratory of Electromagnetic Space Cognition and Intelligent Control, Beijing, 100191, China.}
}

\markboth{Journal of \LaTeX\ Class Files,~Vol.~18, No.~9, April~2024}%
{How to Use the IEEEtran \LaTeX \ Templates}

\maketitle

\begin{abstract}
Accurately detecting the transient signal of interest from the background signal is one of the fundamental tasks in signal processing. The most recent approaches assume the existence of a single background source and represent the background signal using a linear dynamical system (LDS). This assumption might fail to capture the complexities of modern electromagnetic environments with multiple sources. To address this limitation, this paper proposes a method for detecting the transient signal in a background composed of an unknown number of emitters. The proposed method consists of two main tasks. First, a Bayesian nonparametric model called the infinite factorial linear dynamical system (IFLDS) is developed. The developed model is based on the sticky Indian buffet process and enables the representation and parameter learning of the unbounded number of background sources. This study also designs a parameter learning method for the IFLDS using slice sampling and particle Gibbs with ancestor sampling. Second, the finite moving average (FMA) stopping time is introduced to minimize the worst-case probability of missed detection, and the statistical performance of the stopping time is investigated. To facilitate the computation of the FMA stopping time, this study derives the factorial Kalman forward filtering (FKFF) method and designs a dependence structure for the underlying model, allowing the stopping time to be defined by a recursive function. Numerical simulations demonstrate the effectiveness of the proposed method and the validity of the theoretical results. The experimental results of the pulse signal detection under the condition of communication interference confirm the effectiveness and superiority of the proposed method.
\end{abstract}


\section{Introduction}
\IEEEPARstart{M}{odern} electromagnetic environments have become highly complicated and characterized by an increasing number of radiation sources and enhanced dynamic properties, which pose a great challenge to signal detection. This study considers the problem of detecting the transient signal of interest (SOI) from complicated background signals (CBS). The transient signal refers to a signal that begins at an unknown time and has a finite duration. The CBS consists of signals generated by multiple radiation sources\footnote{Signal detection aims to identify the \textit{signal of interest} (emitted by the \textit{source of interest}) from the \textit{background signal} (transmitted by the \textit{background source}).} and environmental noise. We refer to background signal (BS) as the case exist a single background source, distinguishing it from the term CBS.

Traditional signal detection methods, including energy detection~\cite{urkowitz1967energy} and the quickest change detection~\cite{tartakovsky2014sequential}, mostly rely on two assumptions: 1) the BS samples are independent and identically distributed (i.i.d.), and 2) the SOI has infinite duration. However, the i.i.d. background models often fail to represent the time-correlated CBS, which increases the probability of a false alarm and the burden on a downstream signal processing procedure. In addition, when the signal duration is finite, particularly if it is short, the existing detection methods would result in missed detections. Therefore, it is essential to develop a new model that can effectively represent the CBS and define a stopping time using a recursive function.

From the perspective of CBS modeling, it is necessary to alleviate the restriction of the i.i.d. assumption of the signal samples to prevent it from failing to capture the commonly encountered time-correlated BS \cite{sadler1996detection, yang2004memoryless, Kay1993, poor1979memoryless}. Recently, with the development of machine learning, linear dynamical system (LDS)~\cite{ford2019unknown}, hidden Markov model (HMM)~\cite{ford2021unknown}, and switching linear dynamical systems~\cite{ford2023unknown}, have been developed to represent the BS. However, these models assume that a single background source exists, which limits their representation capability in scenarios with multiple background sources. Bayesian non-parametric (BNP) prior offers a potential solution to improve the CBS representation by allowing an infinite number of background sources. The Markov Indian buffet process (MIBP)~\cite{griffiths2011indian}, \cite{teh2007stick} is a common choice to construct the BNP prior, and it has been introduced to factorial HMM~\cite{ghahramani1995factorial} to describe potentially infinite binary Markov chains~\cite{gael2008infinite}. Later, the infinite factorial dynamic model was developed~\cite{valera2015infinite} to enhance the representational capability of the real-life time series. However, none of these BNP models consider the adaptability of electromagnetic signals. Thus, it is necessary to develop a representation model that can automatically determine the number of background sources.

In terms of defining the stopping time for a signal that has a finite duration, the transient change detection (TCD)~\cite{guepie2012sequential} can provide a promising framework. For a more appropriate terminology, this study will refer to the TCD as transient signal detection (TSD) in the rest of this paper. As discussed in~\cite{guepie2012sequential}, there are two types of TSD problems. The first type involves the detection of suddenly arriving signals with random duration. Specifically, Tartakovsky~\cite{tartakovsky2021optimal} provided an optimal solution when the signal duration is geometrically distributed. The second type of TSD problem relates to safety-critical applications where the minimal signal duration or a tolerable detection delay $w_d$ is pre-specified, and the detection delay is greater than $w_d$ is considered a missed detection. Examples of security-critical systems are navigation monitoring systems~\cite{bakhache2000reliable}, cyber-attack detection systems~\cite{fillatre2017security}, and radar signal detection systems~\cite{zhang2023jdmr}. This study considers the second type of TSD problem. The TSD in the general non-i.i.d. background is challenging due to two computational issues. First, the likelihood computation of the factorial state space model is an NP-hard problem~\cite{ghahramani1995factorial}. In~\cite{ghahramani1995factorial}, the authors used the dynamic programming technique to reduce the computational complexity of factorial HMM. To the best of the authors' knowledge, there has been no research on the likelihood-related computational problem of factorial state-space models besides factorial HMM. Second, the log-likelihood ratio (LLR) of non-i.i.d. models is computationally infeasible because the signal arrival time is unknown. For general non-i.i.d. models, a common approach for computing LLR is to assume that the SOI and the BS are independent~\cite{sun2024quickest,Fuh2015,fuh2020asymptotically}. However, this assumption is unsuitable for the scenario considered in this work because the observations depend on the CBS regardless of whether the signal is present or not. These challenges make the existing TSD strategies and theoretical results not suitable for our scenario.


Considering the aforementioned challenges, this study proposes a novel method for transient signal detection in CBS. The proposed method includes two tasks: the CBS Representation and Parameter Learning (RPL) task, and the Transient Signal Detection (TSD) task. The \textbf{RPL} task consists of two components: 1) A model termed the infinite factorial linear dynamical systems (IFLDS) is designed to represent the CBS. In the IFLDS, the sticky MIBP is designed, where sticky control is incorporated into MIBP \cite{griffiths2011indian} to describe the time characteristic of BS. 2) A parameter learning method based on slice sampling (SS) and particle Gibbs with ancestor sampling (PGAS) is developed. The parameter learning method can estimate the parameters of the proposed model including automatically determining the number of background sources. In the \textbf{TSD} task, the finite moving average (FMA) stopping time is introduced. To address the computational problems of the FMA stopping time in the proposed model, this study first proposes a likelihood computation method called factorial Kalman forward filtering (FKFF) inspired by the Kalman Filter \cite{khodarahmi2023review}; then, the LLR is reformulated by designing the dependence structure between the underlying model when SOI is present and absent.

The main contributions of the work can be summarized as follows:
\begin{enumerate}
    \item An IFLDS model is designed to represent CBS, and the sticky control is incorporated to describe the time characteristic of the CBS. The proposed model is fully conjugate and allows for an efficient parameter updating process; 
    \item A parameter learning method based on the SS and PGAS is developed. Such a method can automatically determine the background source number. The simulation results demonstrate the effectiveness and the superiority of the proposed method in terms of performance;
    \item A stopping time is defined based on the FMA test to detect the transient signal. The stopping time is defined by a recursive function, which is suitable for online processing.
    \item The statistical performance of the stopping time is investigated. The relationship between the signal duration and the detection performance is established. The simulations are presented to evaluate the effectiveness of these results.
\end{enumerate}

The remainder of the paper is organized as follows. Section II introduces the formulations of the RPL and the TSD tasks. Section III describes the proposed IFLDS and the corresponding parameter learning method. Section IV defines the FMA stopping time and its statistical performance. Section V presents the simulation design,
results, and discussions. Section VI provides the experimental results of the pulse signal detection in the presence of communication interference. Finally, Section VII concludes this paper.

\section{Problem Formulation}
Assume that multiple sources are recorded by a sequentially observed signal $\boldsymbol{p}=\{\boldsymbol{p}_t\}_{t\geq1}$. At an unknown time $\nu$, the SOI arrives and changes the distribution of the observed signal samples. The SOI lasts $w$ samples, so $\{\boldsymbol{p}_t\}_{t=1}^{\nu-1}$ and $\{\boldsymbol{p}_t\}_{t=\nu+w+1}^{T}$ are generated by one stochastic model and $\{\boldsymbol{p}_t\}_{t=\nu}^{\nu+w}$ is generated by another model. Then, the TSD problem can be defined as a binary hypothesis test as follows:
\begin{equation}
    \label{hypothesis test}
    \begin{aligned}
            &\mathcal{H}_0:\boldsymbol{p}_t = \boldsymbol{n}_t, &\text{if $t<\nu$ or $t\geq\nu+w$},\\
            &\mathcal{H}_1:\boldsymbol{p}_t = \boldsymbol{y}_t +\boldsymbol{n}_t, &\text{if $\nu \leq t \leq \nu+w$},
    \end{aligned}
\end{equation}
\noindent where $\boldsymbol{p}_t=[p_{I,t}, p_{Q,t}]^\top$ is the observed complex signal, $\boldsymbol{n}_t = [n_{I,t},n_{Q,t}]^\top$ is a complex CBS that always presents, and $\boldsymbol{y}_t=[y_{I,t}, y_{Q,t}]^\top$ is the complex deterministic SOI.

In the detection scheme, the number of background sources and the corresponding parameters of the CBS are estimated offline. Once the parameters of the CBS are obtained, the signal emitted by the source of interest can be detected in real-time. The described scenario, featuring two background sources, is illustrated in Fig.~\ref{framework}. In the following two subsections, the problems of the RPL and TSD tasks are formulated.
\begin{figure}[!t]

  \centering
  \includegraphics[width=3.6in]{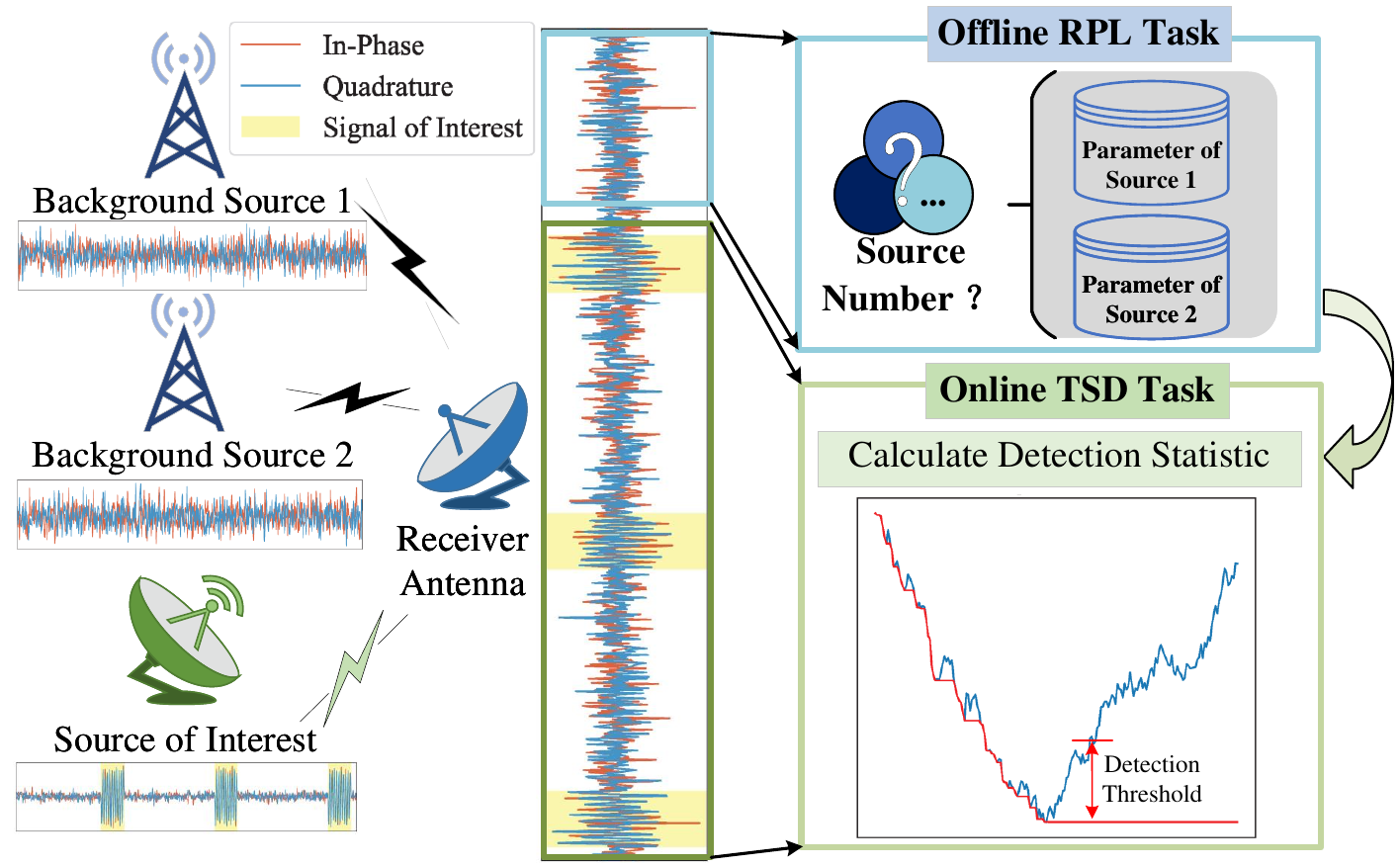}
  \caption{The diagram of the RPL and TSD tasks in the presence of two background sources. The RPL task is performed offline, whereas the TSD task is performed online.
  \label{framework}}
\end{figure}

\subsection{The RPL Task}
In this study, the CBS consists of signals transmitted by $M$ sources. The signal emitted by the $m$th source is generated by an LDS parameterized by $\boldsymbol{\Gamma}^m = \{\boldsymbol{G}^m,\boldsymbol{C}^m, \boldsymbol{R},\boldsymbol{Q}\}, m\in[1,M]$, the observation samples of CBS $\boldsymbol{n}_t$ are drawn from a  probabilistic model described as follows:
\begin{equation}
    \begin{aligned}
    \boldsymbol{x}_1& \sim \mathcal{N}(\boldsymbol{0}, \boldsymbol{Q}),\\
        \boldsymbol{x}_{t+1}^m &= \boldsymbol{G}^m\boldsymbol{x}_{t}^m +\boldsymbol{\omega},\boldsymbol{\omega}\sim \mathcal{N}(\boldsymbol{0},\boldsymbol{Q}),\\
        \boldsymbol{n}_t &= \sum_{m=1}^M(\boldsymbol{C}^m\boldsymbol{x}_t^m + \boldsymbol{v}),\boldsymbol{v}\sim\mathcal{N}(\boldsymbol{0},\boldsymbol{R}),
    \end{aligned}
    \label{FLDS}
\end{equation}
\noindent where $\boldsymbol{G}^m$ and $\boldsymbol{C}^m$ are the transition and output matrices with a dimension $2\times2$ by assumption; $\boldsymbol{\omega}$ and $\boldsymbol{v}$ are zero mean Gaussian random variables with covariance $\boldsymbol{Q}$ and $\boldsymbol{R}$. In this study, this generative model is referred to as the factorial linear dynamical systems (FLDS), and its graphical representation is shown in Fig.~\ref{fig_FLDS}. 
\begin{figure}[!t]

  \centering
  \includegraphics[width=2.5in]{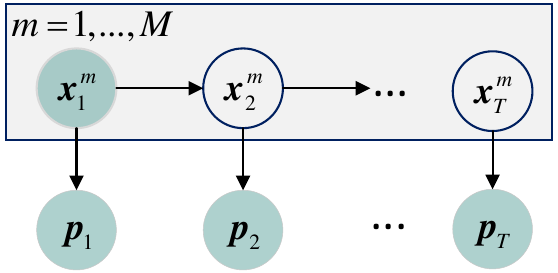}
  \caption{The graphical representation of the FLDS. (The colored nodes denote known variables, and the others represent unknown variables.)
  \label{fig_FLDS}}
\end{figure}

The parameter learning of the FLDS consists of estimating of FLDS parameters $\{\boldsymbol{\Gamma}^m\}_{m=1}^M$ and determining the background source number $M$. The parameter estimation procedure can be formulated as a maximum a posterior problem:
\begin{equation}
    \{\boldsymbol{\Gamma}^m\}_{m=1}^M, M = \arg \max_{\{\boldsymbol{\Gamma}^m\}_{m=1}^M, M} P\big(\{\boldsymbol{\Gamma}^m\}_{m=1}^M,M\mid \{\boldsymbol{n}_t\}_{t=1}^{\bar{T}}\big),
\end{equation}
\noindent where $\bar{T}$ is the length of the CBS used in the RPL task.

\subsection{The TSD Task}
Assuming that all background sources are continuously active and transmitting signals, we focus on detecting an SOI that has a finite duration. Unlike traditional quickest change detection (QCD) problems that focus on minimizing the detection delay, this paper assumes that a tolerable detection delay $w_d$ is equivalent to the signal duration $w$\footnote{The reason for adopting this assumption is that detection with a delay greater than $w_d$ is considered a missing detection, and detection with a delay smaller than $w_d$ is considered less important in terms of the detection performance.}. Next, let $P_\nu$ is the probability measure when $\boldsymbol{p}_1,...,\boldsymbol{p}_{\nu-1}$ and $\boldsymbol{p}_{\nu+m},...,\boldsymbol{p}_T$ are generated by a stochastic model with a marginal probability density function (PDF) $f_0$, corresponding to the null hypothesis (i.e. $\mathcal{H}_0$), on the other hand, $\boldsymbol{p}_\nu,..., \boldsymbol{p}_{\nu+m-1}$ are generated by a stochastic model with PDF of $f_1$ (the signal present case, $\mathcal{H}_1$); $P_\infty$ refers to the case $\nu=\infty$ (i.e., there are no signals). 

The optimization criterion is to minimize the worst-case probability of missed detection:
\begin{equation}
    \inf_{\tau\in C_\alpha}\bigg\{P_{MD}(\tau,w)=\sup_{\nu\geq 1}P_\nu(\tau\geq \nu+w|\tau\geq\nu)\bigg\}
    \label{inf_PMD}.
\end{equation}
where $\inf$ refers to infimum, and $\sup$ is supremum. All stopping times $\tau\in C_\alpha$ satisfy the following condition:
\begin{equation}
    C_{\alpha} = \bigg\{\tau:P_{FA}(\tau,w_\alpha)=\sup_{l\geq 1}P_\infty(l\leq \tau<l+w_\alpha)\leq \tilde{\alpha} \bigg\}
    \label{false alarm constraint},
\end{equation}
\noindent where $P_{FA}$ and $P_{MD}$ denote the worst-case probability of false alarm within any interval with a length of $w_\alpha$ and miss detection, respectively; $l$ is an integer greater than zero; $\tilde{\alpha}\in(0, 1)$ is a pre-defined constant value. 

\section{Representation and Parameter Learning}
This section first describes a conjugate graphical model IFLDS to represent the CBS. Then, the parameter learning algorithm for IFLDS is presented.

\subsection{The Graphical Model}
To estimate the number of background sources, this study uses the LDS~\cite{Bishop2006} and the IFHMM~\cite{gael2008infinite} as two basic building blocks for IFLDS. The graphical representation is shown in Fig.~\ref{sticky_IFLDS}. We present the BNP prior and the joint likelihood function as follows.

\subsubsection{ Bayesian Nonparametric Prior Construction}
In a highly crowded electromagnetic environment, it is common that there exists multiple background sources and the number is unknown. The MIBP~\cite{gael2008infinite} is one of the potential solutions to represent infinite number of background sources. The MIBP places a prior distribution over a binary matrix $\boldsymbol{S}$ with a finite number of rows and an infinite number of columns. In this matrix, the $t$th row represents a time step $t$, and the $m$th column indicates the binary states of the $m$th Markov chain. Each element $s_t^m \in \{0,1\}$ indicates whether the \textit{m}th LDS is active at a time moment $t$, where $t\in[1,\bar{T}]$ is the time index of the Markov chain (i.e., the row number of the $\boldsymbol{S}$), and $m\in[1,M], M\to \infty$ is the source index of the Markov chain. This study assumes $s_0^m=0$ for simplicity. The variable $s_t^m$ evolves according to the transition matrix as follows:
\begin{equation}
    \boldsymbol{A}^m = 
    \begin{pmatrix}
        1-a^m & a^m\\1-b^m & b^m
    \end{pmatrix},
    \label{old transition}
\end{equation}
\noindent where $a^m = P(s_t^m=1|s_{t-1}^m=0)$ and $b^m=P(s_t^m=1|s_{t-1}^m=1)$. 

The MIBP is implemented using the stick-breaking construction \cite{gael2008infinite}, which facilitates simple and efficient inference. Denote $a^{(m)}$ as a sorted value of $a^m$, which satisfies the condition of $a^{(1)} > ...>a^{(m)}>...$, and it holds that:
\begin{equation}
    \begin{aligned}
        a^{(1)}&\sim Beta(\alpha,1),\\
        P(a^{(m)}|a^{(m-1)})&\propto (a^{(m)})^{\alpha-1}\mathbb{I}(0\leq a^{(m)}\leq a^{(m-1)}),
    \end{aligned}
    \label{sampleam}
\end{equation}
\noindent where $Beta(\cdot)$ is the Beta function, $\mathbb{I}(\cdot)$ is the indicator function, and $\alpha$ is the concentration parameter of the BNP prior.  The prior distribution placed over variable $b^m$ is the Beta distribution:
\begin{equation}
    b^m\sim Beta(\beta_0,\beta_1).
\end{equation}
where $\beta_0$ and $\beta_1$ are hyper-parameters. The prior distribution implies that only a finite number of Markov chains is active, whereas the others remain idle. To better describe the temporal characteristic of the continuously active background sources, in the following part, we introduce the sticky control.

In electromagnetic environments, the background sources, such as communication transmitters, typically do not switch on and off frequently within a short time interval. This suggests that the hidden state, denoted by $\boldsymbol{s}_t^m$, tends to persist in the current state rather than to move to other states. To describe such a pattern, this study introduces a sticky control to govern the transition between adjacent hidden states of $\boldsymbol{S}^m$. The sticky control serves as a valuable tool for controlling the Markov process and has been applied in many fields, including speaker diarization~\cite{Fox2011}, signal detection~\cite{ford2023unknown}, and comparative genomic hybridization~\cite{Du2010sticky}. The sticky control provides an effective approach for accurately modeling temporal characteristics. 

In this study, sticky control is implemented by introducing $\bar{T}$ random variables to each Markov chain $z_t^m \in \{0,1\}$. The transition density of a hidden variable $s_t^m$ is conditioned on a sticky variable as follows:
\begin{equation}
\begin{aligned}
P(s_t^m=i|s_{t-1}^m=j,z_t^m,\boldsymbol{A}^m)&=
\begin{cases}
           \boldsymbol{A}^m_{j,i},& \text{if $z_t^m=1$},\\
            \boldsymbol{I}_{j,i}, & \text{if $z_t^m=0$},
\end{cases}
\end{aligned}
\label{new transition}
\end{equation}
\noindent where $\boldsymbol{I}_{j,i}$ is the element in the $j$th row and the $i$th column of the dimensional identity matrix, and $\boldsymbol{A}_{j,i}^m$ is the element in the $j$th row and the $i$th column of the matrix $\boldsymbol{A}^m$.

A sticky variable $z^m_t$ is sampled from the Bernoulli distribution, with a Beta prior distribution parametrized by $\gamma_0$ and $\gamma_1$: 
\begin{equation}
    \begin{aligned}
        P(z^m_t)&= Ber(\gamma^m),\\
        P(\gamma^m)& = Beta(\gamma_0,\gamma_1),
    \end{aligned}
    \label{betaprior}
\end{equation}
where $Ber(\cdot)$ is the Bernoulli distribution.

\subsubsection{Joint Probability Density Function}
This section first defines the notations of random variables as follows: $\boldsymbol{X}=\{\boldsymbol{X}^m\}_{m\in[1,M]}=\{\boldsymbol{x}_t^m\}_{t\in[1,T],m\in[1,M]}, \boldsymbol{P}=\{\boldsymbol{p}_t\}_{t\in[1,T]}, \boldsymbol{A}=\{\boldsymbol{A}^m\}_{m\in[1,M]}, \boldsymbol{\Gamma} = \{\boldsymbol{\Gamma}^m\}_{m\in[1,M]},\boldsymbol{S}=\{\boldsymbol{S}^m\}_{m\in[1,M]}=\{s_t^m\}_{t\in[1,T],m\in[1,M]},\boldsymbol{Z}=\{\boldsymbol{Z}^m\}_{m\in[1,M]}=\{z_t^m\}_{t\in[1,T],m\in[1,M]}$. The $m$th LDS is active at a time $t$ and contributes its value to the observation if and only if $s_t^m$ takes a value of one. Meanwhile, the $m$th latent Markov chain, $\boldsymbol{S}^m$, will persist in its current state if $z_t^m=1$. The joint likelihood function of the IFLDS is defined by:
\begin{equation}
\label{likelihood}
\begin{aligned} &P(\boldsymbol{S},\boldsymbol{X},\boldsymbol{P},\boldsymbol{A},\boldsymbol{Z},\boldsymbol{\Gamma})\propto\prod_{t=1}^{\bar{T}} \prod_{m=1}^{M} P(s_t^m \mid s_{t-1}^m, z_t^m, \boldsymbol{A}^m) \times \\
&\prod_{t=1}^{\bar{T}} \prod_{m=1}^{M} P(\boldsymbol{x}_t^m \mid \boldsymbol{x}_{t-1}^m, \boldsymbol{G}^m, s_t^m) \times \\
&\prod_{t=1}^{\bar{T}} P(\boldsymbol{p}_t \mid \{\boldsymbol{x}_t^m,\boldsymbol{\Gamma}^m\}_{m=1}^M)\times P(\boldsymbol{A})\times P(\boldsymbol{Z})\times P(\boldsymbol{\Gamma}).
\end{aligned}
\end{equation}
The first term of the joint likelihood function \eqref{likelihood} describes the hidden state transition rule of Markov chain defined by~\eqref{new transition}; the second term is the transition rule of the LDS, which satisfies the following expression:
\begin{equation}
    P(\boldsymbol{x}_t^m|\boldsymbol{x}_{t-1}^m,\boldsymbol{G}^m,s_t^m)=
    \begin{cases}
            \mathcal{N}(\boldsymbol{G}^m\boldsymbol{x}_{t-1}^m,\boldsymbol{Q}),& \text{if $s_t^m=1$},\\
            \mathcal{N}(\boldsymbol{x}_{t-1}^m,\boldsymbol{Q}), & \text{if $s_t^m=0$}.
    \end{cases}
    \label{IFLDS}
\end{equation}

The third term of~\eqref{likelihood} describes the observation model; an observation $\boldsymbol{p}_t$ is the sum of the outputs of all LDSs, and it follows $\boldsymbol{p}_t=\mathcal{N}(\sum_{m=1}^M \boldsymbol{C}^m\boldsymbol{x}_t^m,M\boldsymbol{R})$. The fourth and the fifth terms of~\eqref{likelihood} are described by~\eqref{sampleam} and~\eqref{betaprior}, respectively. 

The last term of~\eqref{likelihood} represents the prior distribution of each LDS: $P(\boldsymbol{\Gamma}) = \prod_{m=1}^M P(\boldsymbol{\Gamma}^m)= \prod_{m=1}^M P(\boldsymbol{G}^m, \boldsymbol{Q})P(\boldsymbol{C}^m,\boldsymbol{R})$. To ensure conjugacy, this study places the matrix normal inverse Wishart (MNIW) prior to the LDS transition and output matrices. Particularly, for the latent variables' parameters $\{\boldsymbol{G}^m,\boldsymbol{Q}\}$, the latent relation variable $\boldsymbol{D}^m=\{\boldsymbol{\psi}^m, \bar{\boldsymbol{\psi}}^m\}$ is defined, where $\boldsymbol{\psi}^m_t = \boldsymbol{x}_t^m$ and $\boldsymbol{\bar{\psi}}^m_t = \boldsymbol{x}_{t-1}^m$. The posterior distribution can be expressed using the Bayesian theorem as follows:
\begin{equation}
\begin{aligned}
    P(\boldsymbol{G}^m, \boldsymbol{Q}|\boldsymbol{D}^m) &= P(\boldsymbol{G}^m|\boldsymbol{Q}, \boldsymbol{D}^m)P(\boldsymbol{Q}|\boldsymbol{D}^m),\\
    P(\boldsymbol{G}^m|\boldsymbol{Q}) &= \mathcal{MN}(\boldsymbol{G}^m;\boldsymbol{M}_0,\boldsymbol{Q},\boldsymbol{K}_0),\\
    P(\boldsymbol{Q})&=\mathcal{IW}(n_0,\boldsymbol{S}_0),
\end{aligned}
\label{prioroftransition}
\end{equation}
\noindent where $\mathcal{MN}$ is the matrix normal distribution, $\mathcal{IW}$ is the inverse Wishart distribution, and $\boldsymbol{M}_0,\boldsymbol{K}_0,n_0$, and $\boldsymbol{S}_0$ are the hyper-parameters. The prior distribution of the output parameters $\boldsymbol{C}^m$ and $\boldsymbol{R}$ are normal inverse Wishart distribution with $\boldsymbol{M}_0,\boldsymbol{K}_0,n_0$, and $\boldsymbol{S}_0$ as its hyper-parameters.

\begin{figure}[!t]

  \centering
  \includegraphics[width=3.5in]{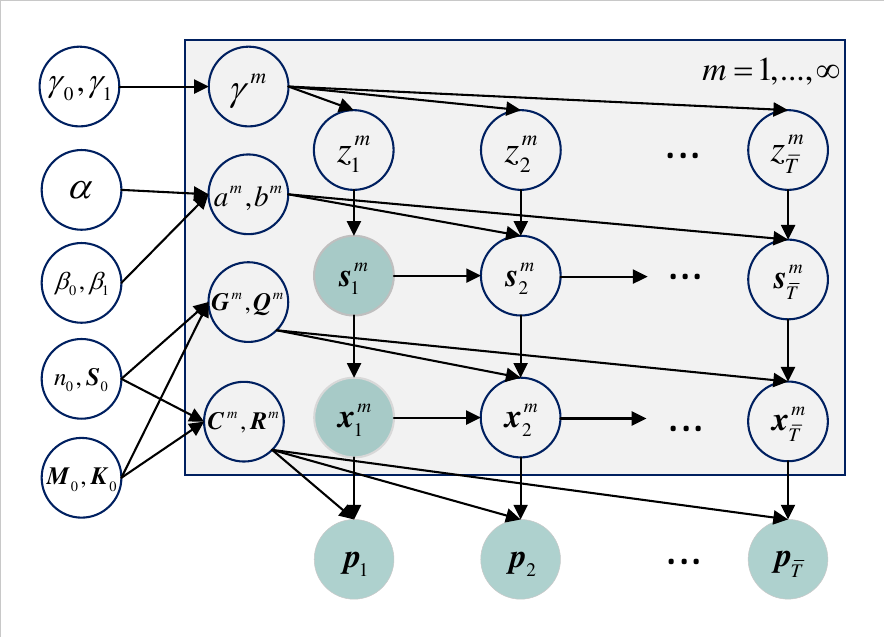}
  \caption{The graphical representation of the IFLDS; variable $\boldsymbol{x}_t^m$ denotes the latent variable of the $m$th background source at a time instant $t$; variable $s_t^m$ indicates whether the $m$th source is active or idle at time instant $t$; variable $z_t^m$ indicates whether $s_t^m$ persist on its state.
  }
  \label{sticky_IFLDS}
\end{figure}
\subsection{Parameter Learning}
This paper combines slice sampling (SS) \cite{neal2003slice} and particle Gibbs with ancestor sampling (PGAS) \cite{lindsten2014particle} to learn the parameters of the IFLDS. A similar approach was adopted in~\cite{gael2008infinite}, \cite{valera2015infinite}, \cite{ ruiz2018infinite}. The learning process in the proposed model consists of three iterative steps: sampling the number of parallel chains ($M$), sampling local parameters $\boldsymbol{S},\boldsymbol{X},\boldsymbol{Z}$ conditioned on the global parameters $\boldsymbol{\Gamma}$, and sampling the global parameters given the local parameters. These steps are explained in detail below. 

\subsubsection{Parallel Chain Number Sampling}
To sample the number of parallel chains, this study employs the SS method \cite{gael2008infinite}. By assuming that there are $M'$ active chains, an auxiliary slice variable $\vartheta$ is sampled through a uniform distribution $\mathcal{U}$ as follows:
\begin{equation}
    \vartheta \sim \mathcal{U}(0, \min_{m:\exists t,s_t^m=1}a^{(m)}).
\end{equation}
A new Markov chain with all idle states $\boldsymbol{s}_t^{M'+1}=0, t\in[1,\bar{T}]$ is introduced to the model. The corresponding hidden states $z_t^{M'+1}$ are all set to one. Then, the variable $a^{(m)},m\in(1,M'+1)$ is iteratively sampled according to the following equation:
\begin{equation}
\begin{aligned}
     P(a^{(m)}|a^{(m-1)})\propto \exp \bigg(\alpha \sum_{t=1}^{\bar{T}} \frac{1}{t} (1-a^{(m)})^t\bigg) \\
     \times(a^{(m)})^{\alpha-1}(1-a^{(m)})^\top \mathbb{I}(0\leq a^{(m)}\leq a^{(m-1)}).
\end{aligned}
\end{equation}
where $a^{(M')}=\min_{m:\exists t,s_t^m=1}a^{(m)}$. If $a^{(m)} $ exceeds the slice variable, matrices $\boldsymbol{S}$, $\boldsymbol{X}$, and $\boldsymbol{Z}$ are expanded by adding a new column, and the global variables $\boldsymbol{\Gamma}^{M'+1}$ of the new Markov chain are sampled using the prior distribution defined by~\eqref{prioroftransition}.
\subsubsection{Local Variables Sampling}
After sampling the number of parallel chains, the local parameters of finite chains need to be estimated. To this end, this study adopts the PGAS, where a sequential Monte Carlo (SMC) is running with one particle, and it is set deterministically to a reference input particle. This reference particle corresponds to the previous iteration's output. Next, a new reference trajectory is obtained by selecting one of the particle trajectories with probabilities given by their importance weights. 

There are $Q$ particles at each time, representing the hidden states $\{\boldsymbol{x}_t^m\}_{m=1}^{M'}$. Therefore, to simplify the representation, this study defines vector $\boldsymbol{\xi}_t(i)$ with a length $M$ as a state of the $i$th particle at a time $t$ and $\boldsymbol{\xi}_t(Q)$ as the input reference particle. In addition, ancestor indexes $a_t^i\in\{1,...,Q\}$ that correspond to the indexes of the ancestor particles of $\boldsymbol{\xi}_t(i)$ are defined. The $i$th particle trajectory is defined as follows:
\begin{equation}
    \boldsymbol{\xi}_{1:t}(i) = concat\bigg(\boldsymbol{\xi}_{1:t-1}({a_t^i}), \boldsymbol{\xi}_t(i)\bigg).
\end{equation}
where $concat$ is the concatenate operator. At a time $t$, the proposed method first generates the ancestor indexes for the first ($Q-1$) particles according to the importance weight $w_{t-1}^i$ at the previous time, and the $Q$th particle $\boldsymbol{\xi}_t(Q)$ is set to the reference particle. Each particle evolves across time according to the transition rule defined in~\eqref{IFLDS}.

To perform the PGAS, this study first generates the ancestor indexes for the first $Q-1$ particles according to importance weights $w_{t-1}^i$. The particles are propagated across time according to~\eqref{FLDS}.
Particle weight $w_t^i$ and ancestor weights $\tilde{w}_{t-1|\bar{T}}^i$ are respectively obtained by:
\begin{equation}
    \begin{aligned}
        w_t^i  &\propto P(\boldsymbol{p}_t|\boldsymbol{\xi}_t),\\
        \tilde{w}_{t-1|\bar{T}}^i &\propto w_{t-1}^iP(\boldsymbol{\xi}_t(Q)|\boldsymbol{\xi}_{t-1}(i)),
    \end{aligned}
    \label{compute weights}
\end{equation}
\noindent where the transition probability of particles $P(\boldsymbol{\xi}_t|\boldsymbol{\xi}_{t-1}(a_t))$ is defined as follows:
\begin{equation}
\begin{aligned}
        P(\boldsymbol{\xi}_t|\boldsymbol{\xi}_{t-1}(a_t))=\prod_{m=1}^M P(\boldsymbol{x}_t^m|s_t^m)P(s_t^m|z_t^m,{s_{t-1}^m}({a_t})).
\end{aligned}
\end{equation}
where $s_{t-1}^m(a_t)$ is the hidden variable of the $a_t$th particle, and the global variables are dropped for simplicity. For a detailed description of the implementation process and a gentle introduction to the PGAS please refer to~\cite{lindsten2014particle}~and~\cite{gauraha2020note}.
\subsubsection{Global Variables Sampling}
The global variables of the IFLDS include $a^m,b^m,\gamma^m$ and $\boldsymbol{\Gamma}^m$. Parameters $a^m, b^m$ and $\gamma^m$ are sampled from beta distribution according to the stick-breaking construction of the MIBP~\cite{teh2007stick} as follows:
\begin{equation}
\begin{aligned}
        P(a^m|\boldsymbol{S} )&\sim Beta(n_{01}^m,1+n_{00}^m),\\
        P(b^m|\boldsymbol{S}) &\sim Beta(\beta_0+n_{11}^m, \beta_1+n_{10}^m),
\end{aligned}
\end{equation}
\noindent where $n_{ij}^m$ is the number of transitions from a state $i$ to the state $j$ of the $m$th column of $\boldsymbol{S}$, when $z_t^m=1$, where $t$ is the row index of $\boldsymbol{S}$. 

The hyper-parameter of sticky variables $\gamma^m$ is sampled from Beta distribution as follows:
\begin{equation}
    P(\gamma^m|\boldsymbol{Z}) \sim Beta(\gamma_0+n_0^m, \gamma_1+n_1^m),
\end{equation}
where $n_i^m$ is the order number of state $i$ in a state sequence $\boldsymbol{Z}^m$. 
Given the designed conjugate prior, the posterior distribution for the transition of FLDS is expressed as follows:
\begin{equation}
\begin{aligned}
     P(\boldsymbol{G}^m|\boldsymbol{Q},\boldsymbol{D}^m) &= \mathcal{MN}(\boldsymbol{G}^m;\boldsymbol{S}_{\boldsymbol{\psi}\bar{\boldsymbol{\psi}}}^m{\boldsymbol{S}_{\bar{\boldsymbol{\psi}}\bar{\boldsymbol{\psi}}}^m}^{-1},\boldsymbol{Q},\boldsymbol{S}_{\bar{\boldsymbol{\psi}}\bar{\boldsymbol{\psi}}}^m),\\
     P(\boldsymbol{Q}|\boldsymbol{D}^m)&=\mathcal{IW} (n_0+\bar{T}, \boldsymbol{S}_0+\boldsymbol{S}_{\boldsymbol{\psi}|\bar{\boldsymbol{\psi}}}^m),\\
\end{aligned}
\end{equation}
where variables $\boldsymbol{S}_{\bar{\boldsymbol{\psi}}\bar{\boldsymbol{\psi}}}^m$, $\boldsymbol{S}_{\boldsymbol{\psi}\bar{\boldsymbol{\psi}}}^m$, $\boldsymbol{S}_{\boldsymbol{\psi\psi}}^m$, and $\boldsymbol{S}_{\boldsymbol{\psi}|\bar{\boldsymbol{\psi}}}^m$ are defined as follows:
\begin{equation}
\begin{aligned}
    \boldsymbol{S}_{\bar{\boldsymbol{\psi}}\bar{\boldsymbol{\psi}}}^m &= \bar{\boldsymbol{\psi}}^m (\boldsymbol{\psi}\bar{\boldsymbol{\psi}}^m)^\top + \boldsymbol{K}_0,\\
    \boldsymbol{S}_{\boldsymbol{\psi}\bar{\boldsymbol{\psi}}}^m &= \boldsymbol{\psi}^m (\bar{\boldsymbol{\psi}}^m)^\top + \boldsymbol{M}_0\boldsymbol{K}_0,\\
\boldsymbol{S}_{\boldsymbol{\psi}\boldsymbol{\psi}}^m &= \boldsymbol{\psi}^m (\boldsymbol{\psi}^m)^\top + \boldsymbol{M}_0\boldsymbol{K}_0\boldsymbol{M}_0^\top,\\
    \boldsymbol{S}_{\boldsymbol{\psi}|\bar{\boldsymbol{\psi}}}^m &= \boldsymbol{S}_{\boldsymbol{\psi}\boldsymbol{\psi}}^m{\boldsymbol{S}_{\bar{\boldsymbol{\psi}}\bar{\boldsymbol{\psi}}}^m }^{-1}\boldsymbol{S}_{\boldsymbol{\psi}\bar{\boldsymbol{\psi}}}^m.
\end{aligned}
\end{equation}
The posterior distribution of the output parameter of FLDS $\boldsymbol{C}^m$ and $\boldsymbol{R}$ are expressed similarly.

\section{Transient Signal Detection and Statistical Performance}
The IFLDS reduces to the FLDS once the number of the background source ($M$) are estimated, as shown in Fig.~\ref{FLDS}. However, finding an optimal solution to the TSD problem for $1<w<\infty$ has still been an open problem in the literature. The Window-Limited CUSUM (WL-CUSUM) test~\cite{guepie2012sequential},~\cite{lai1998information} suggested to use windowed samples to calculate the detection statistic. In~\cite{guepie2012sequential}, it has been shown that after certain optimization, the WL-CUSUM for Gaussian mean change is equivalent to the FMA test. Reference~\cite{egea2018performance} has also shown that the FMA test provides a tighter bound compared to the WL-CUSUM test~\cite{egea2018performance}, allowing for a more precise prediction for detection performance. Accurate performance prediction is essential for a safe-critical system. Based on the above investigations, this study defines the stopping time using the FMA test. A detailed comparison of TSD problems can be found in~\cite{egea2022two}. The FMA detection statistic is defined as follows:
\begin{equation}
    \begin{aligned}
            W_t&= \max_{t-w+1<\nu<t}\sum_{i=t-w+1}^t \log \frac{P_\nu(\boldsymbol{p}_t|\{\boldsymbol{p}_i\}_{i=t-w+1}^{t-1})}{P_\infty(\boldsymbol{p}_t|\{\boldsymbol{p}_i\}_{i=t-w+1}^{t-1})}.
    \end{aligned}
    \label{GLR-FMA}
\end{equation}

There are two computational challenges in the above detection statistic. First, the computation of the data likelihood under the null hypothesis $P_\infty(\boldsymbol{p}_t|\{\boldsymbol{p}_i\}_{i=t-w+1}^{t-1})$ is invalid due to the factorial structure \cite{ghahramani1995factorial}. Second, the computation of the LLR is infeasible because the exhaustive search of a signal arrival time $\nu$ can not written in a recursive manner, which causes a computational burden to the detector. These computational problems have been investigated within the QCD framework. Fuh~\cite{Fuh2015} proposed the recursive CUSUM algorithm for the HMM, followed by the Shiryayev-Roberts-Pollak (SRP) algorithm for the state
space models~\cite{fuh2020asymptotically}. However, the TSD problem in the FLDS is still an open problem.

The following subsections first propose the FKFF to compute the data likelihood under the null hypothesis $P_\infty(\boldsymbol{p}_t|\{\boldsymbol{p}_i\}_{i=t-w+1}^{t-1})$ recursively. Then, the dependence structure of the underlying model is designed to address the computational problem of the LLR. Two computational challenges are addressed. Finally, the statistical performance of the proposed FMA stopping time is analyzed.

\subsection{Factorial Kalman Forward Filtering}
This study recursively calculates data likelihood $P_\infty(\boldsymbol{p}_t|\{\boldsymbol{p}_i\}_{i=t-w+1}^{t-1})$ under the null hypothesis, which is crucial for online processing. This section defines the forward variable for the FLDS as follows: 
\begin{equation}
\begin{aligned}
        &\boldsymbol{\alpha}_t\bigg(\{\boldsymbol{x}_t^m\}_{m=1}^M\bigg) \\
        =& P\bigg(\{\boldsymbol{p}_i\}_{i=1}^t,\{\boldsymbol{x}_t^m\}_{m=1}^M\bigg)\propto P\bigg(\{\boldsymbol{x}_t^m\}_{m=1}^M|\{\boldsymbol{p}_i\}_{i=1}^t\bigg)\\
         \propto&\int \boldsymbol{\alpha}_{t-1}\big(\{\boldsymbol{x}_{t-1}^m\}_{m=1}^M\big)P\big(\{\boldsymbol{x}_{t}^m\}_{m=1}^M|\{\boldsymbol{x}_{t-1}^m\}_{m=1}^M \big)\\
    &P(\boldsymbol{p}_{t}|\{\boldsymbol{x}_{t}^m\}_{m=1}^M)d\{\boldsymbol{x}_{t-1}^m\}_{m=1}^M.
\end{aligned}
\end{equation}

It is convenient to separate the calculation process of the forward variable $\boldsymbol{\alpha}_t$ into two iterative steps: calculation of the prediction variable and update variable, these intermediate variables are defined as:

1) Prediction variable:
\begin{equation}
    \begin{aligned}
        &\bar{\boldsymbol{\alpha}}_{t-1}(\{\boldsymbol{x}_{t}^m\}_{m=1}^M)\\
        =&\hat{\boldsymbol{\alpha}}_{t-1}(\{\boldsymbol{x}_{t-1}^m\}_{m=1}^M)P\big(\{\boldsymbol{x}_{t}^m\}_{m=1}^M|\{\boldsymbol{x}_{t-1}^m\}_{m=1}^M \big);\\
    \end{aligned}
    \label{predictionvariable}
\end{equation}

2) Update variable:
\begin{equation}
    \begin{aligned}
        &\boldsymbol{\hat{\alpha}}_{t}(\{\boldsymbol{x}_{t}^m\}_{m=1}^M)\\
        =&\int \bar{\boldsymbol{\alpha}}_{t-1}(\{\boldsymbol{x}_{t-1}^m\}_{m=1}^M)P\big(\boldsymbol{p}_t|\{\boldsymbol{x}_{t}^m\}_{m=1}^M \big)d\{\boldsymbol{x}_{t-1}^m\}_{m=1}^M\\
        \propto& \boldsymbol{\alpha}_{t}(\{\boldsymbol{x}_t^m\}_{m=1}^M).
    \end{aligned}
    \label{updatevariable}
\end{equation}

In the following lemma, the forward variable is proven to be a multi-variate Gaussian function of $\{\boldsymbol{x}_{t}^m\}_{m=1}^M$.
\begin{lemma}
    A forward variable $\boldsymbol{\alpha}_{t}\big(\{\boldsymbol{x}_{t}^m\}_{m=1}^M\big)$ is a multi-variate Gaussian function with respect to $\{\boldsymbol{x}_{t}^m\}_{m=1}^M$.
\label{lemmaforwardvariable}
\end{lemma}
\begin{proof}[Proof outline]
Lemma 1 can be proven using induction. When $t=1$, $\boldsymbol{\alpha}_{1}\big(\{\boldsymbol{x}_{1}^m\}_{m=1}^M\big) = P\big(\{\boldsymbol{x}_{1}^m\}_{m=1}^M\big)P\big(\boldsymbol{p}_1|\{\boldsymbol{x}_{1}^m\}_{m=1}^M\big)$ is a multi-variate Gaussian function based on~\eqref{FLDS}. This function can be established by computing the product of two multivariate Gaussian functions. For $t>1$, $\alpha_t\big(\{\boldsymbol{x}_t^m\}_{m=1}^M\big)$ is a Gaussian function. This can be proven by the product and the convolution of Gaussian density functions. This study parametrizes the update and prediction steps of the FLDS as follows:
\begin{equation}
\begin{aligned}
          \hat{\boldsymbol{\alpha}}_{t}\big(\{\boldsymbol{x}_{t}^m\}_{m=1}^M\big)&\sim \mathcal{N}(\hat{\boldsymbol{\mu}}_t, \hat{\boldsymbol{\Sigma}}_t),\\
          \bar{\boldsymbol{\alpha}}_{t}\big(\{\boldsymbol{x}_{t}^m\}_{m=1}^M\big)&\sim\mathcal{N}(\bar{\boldsymbol{\mu}}_t, \bar{\boldsymbol{\Sigma}}_t),
\end{aligned}
\end{equation}
\noindent where $\bar{\boldsymbol{\mu}}_t = col\bigg(\{ \bar{\boldsymbol{\mu}}_t^m\}_{m=1}^M\bigg)$ and $\hat{\boldsymbol{\mu}}_t = col\bigg(\{ \hat{\boldsymbol{\mu}}_t^m\}_{m=1}^M\bigg)$; also $\bar{\boldsymbol{\Sigma}}_t= diag\bigg(\{ \bar{\boldsymbol{\Sigma}}_t^m\}_{m=1}^M\bigg)$ and $\hat{\boldsymbol{\Sigma}}_t= diag\bigg(\{ \hat{\boldsymbol{\Sigma}}_t^m\}_{m=1}^M\bigg)$; operator $col(\cdot)$ generates a column vector using the input variables; operator $diag(\cdot)$ constructs a diagonal matrix with the input variables; thus, the mean value of the Gaussian distributions $\hat{\boldsymbol{\mu}}_t$ and $ \bar{\boldsymbol{\mu}}_t$ are matrices with a size of $(2M, 1)$, and the covariance of the Gaussian distributions $\hat{\boldsymbol{\Sigma}}_t$ and $\bar{\boldsymbol{\Sigma}}_t$ are matrices with a size of $(2M, 2M)$. 

Based on the knowledge about the distribution family, the update function of forward variables can be analytically expressed:
\begin{equation}
\begin{aligned}
        \bar{\boldsymbol{\mu}}_t^m&=\boldsymbol{G}^m \hat{\boldsymbol{\mu}}_{t-1}^m,\\
        \bar{\boldsymbol{\Sigma}}_t^m &= \boldsymbol{G}^m\hat{\boldsymbol{\Sigma}}_{t-1}^m {\boldsymbol{G}^m}^\top +\boldsymbol{Q},\\
        \hat{\boldsymbol{\mu}}_t^m&=\boldsymbol{G}^m \hat{\boldsymbol{\mu}}_{t-1}^m+\boldsymbol{K}^m\big(\boldsymbol{p}_t-\sum_{m=1}^M \boldsymbol{C}^m\boldsymbol{G}^m \boldsymbol{x}_{t-1}^m\big),\\
         \hat{\boldsymbol{\Sigma}}_t^m &= \bar{\boldsymbol{\Sigma}}_{t}^m
-\boldsymbol{K}^m\bigg({\boldsymbol{C}^m}(\bar{\boldsymbol{\Sigma}}_t^m)^\top
+\sum_{\substack{n=1\\n\neq m}}\boldsymbol{C}^n \boldsymbol{Q}^\top\bigg),\\
\end{aligned}
\label{kalman update}
\end{equation}
\noindent where $\boldsymbol{K}^m$ is defined as a factorial Kalman gain that is expressed by:
\begin{equation}
\begin{aligned}
    \boldsymbol{K}^m &= \bigg(\bar{\boldsymbol{\Sigma}}_t^m{\boldsymbol{C}^m}^\top+\sum_{\substack{n=1\\n\neq m}}\boldsymbol{Q}{\boldsymbol{C}^n}^\top\bigg)\\
\bigg(\sum_{n=1}^M &\boldsymbol{C}^n\bar{\boldsymbol{\Sigma}_t^n}{\boldsymbol{C}^n}^\top+\sum_{n=1}^M\sum_{\substack{m=1\\m\neq n}}^M \boldsymbol{C}^n\boldsymbol{Q}{\boldsymbol{C}^m}^\top +M^2\boldsymbol{R}\bigg)^{-1}.
\end{aligned}
\label{kalman gain}
\end{equation}
The detailed derivation of \eqref{kalman update} can be found in Section I of Supplemental Material.
\end{proof}

The main difference between the FKFF and conventional Kalman forward filtering \cite{chopin2020introduction} lies in the update variable calculation, where the correlation between different LDSs is introduced. In the FLDS model, latent variable $\boldsymbol{x}_t^m$, $\boldsymbol{x}_t^n (n\neq m)$ and observed variable $\boldsymbol{p}_t$ construct a v-structure. The latent variables $\boldsymbol{x}_t^m$ and $\boldsymbol{x}_t^n$ are independently conditioned on the observed variable $\boldsymbol{p}_t$. A detailed discussion of the v-structure has been found in~\cite{koller2009probabilistic}.

Given the forward variable, the data likelihood function at a time $t$, conditioned on the previous observations, can be expressed by:
\begin{equation}
P(\boldsymbol{p}_t|\{\boldsymbol{p}_i\}_{i=1}^{t-1})  =\mathcal{N}(\boldsymbol{p}_t;\boldsymbol{\mu}_t^p, \boldsymbol{\Sigma}_t^p),
\label{FFKF conditional likelihood}
\end{equation}
\noindent where $\boldsymbol{\mu}^p_t = \sum_{m=1}^M\big(\boldsymbol{C}^m\bar{\boldsymbol{\mu}}_t^m\big)$ and $\boldsymbol{\Sigma}_t^p = \sum_{m=1}^M\big(\boldsymbol{C}^m\bar{\boldsymbol{\Sigma}}_t^m{\boldsymbol{C}^m}^\top+M\boldsymbol{R}\big)$.  Next, it can be derived that the joint probability density is equal to the sum of the conditional distribution, which can be expressed by:
\begin{equation}
    P(\{\boldsymbol{p}_i\}_{i={t-w+1}}^t) = \prod_{i=t-w+1}^{t} P(\boldsymbol{p}_i|\{\boldsymbol{p}_j\}_{j=1}^{i-1}).
    \label{FKFFbenefit}
\end{equation}

Although the observed data is not i.i.d., the likelihood function can be written as the product of the observation likelihoods, which is the product of Gaussian functions.

\subsection{Finite Moving Average Stopping Time}
This section reformulates the detection statistic by designing a dependence structure of the underlying model when the SOI is present and absent. The designed structure is based on the following assumptions:

\begin{assumption}
The time index $t$ approaches $\infty$.
\end{assumption}
\begin{assumption}
    The latent variable at a time $\nu$ ($\boldsymbol{x}_\nu^m$) is independent of the latent variable at a time $\nu-1$ ($\boldsymbol{x}_{\nu-1}^m$), but it depends on the previous ($\nu-1$) observations.
\end{assumption}
Assumption 1 indicates that the latent variable at a time $t$ is independent of its initial value. Assumption 2 relaxes the dependence between the latent variables $\boldsymbol{x}_\nu^m$ and $\boldsymbol{x}_{\nu-1}^m$. Given these two assumptions, we have the following Theorem:
\begin{theorem}
Given Assumptions 1 and 2. The probability density function of the observation at a time $t$ $P_\nu(\boldsymbol{p}_t|\{\boldsymbol{p}_i\}_{i=t-w+1}^{t-1})$ is equivalent to $P_1(\boldsymbol{p}_t|\{\boldsymbol{p}_i\}_{i=t-w+1}^{t-1})$.
    \label{FMAgap}
\end{theorem}
\begin{proof}
 The proof is given in Section II of the Supplemental Material.
\end{proof}
Given Theorem \ref{FMAgap}, the LLR of observation $\boldsymbol{p}_t$ can be reformulated as follows:
\begin{equation}
    \hat{L}_t=\log \frac{P_1(\boldsymbol{p}_t|\{\boldsymbol{p}_i\}_{i=1}^{t-1})}{P_\infty(\boldsymbol{p}_t|\{\boldsymbol{p}_i\}_{i=1}^{t-1})}.
    \label{likelihood ratio}
\end{equation}

The above LLR can be intuitively understood as two FKFF algorithms running in parallel conditioned on the same observations, one with the SOI never arriving and the other with the signal arriving at $t=1$ and never ends. A similar transformation was performed in  \cite{sun2024quickest}, \cite{Fuh2015}. In~\eqref{likelihood ratio}, the denominator can be obtained by direct application of~\eqref{kalman update},~\eqref{kalman gain},~and~\eqref{FFKF conditional likelihood}. The numerator of~\eqref{likelihood ratio} denotes the likelihood under the alternative hypothesis $\mathcal{H}_1$ with PDF $\mathcal{N}(\boldsymbol{\mu}_t^p + \boldsymbol{y}_t,\boldsymbol{\Sigma}_t^p)$. Then, the FMA stopping time $\tau$ can be written as follows:
\begin{equation}
    \tau = \inf \{t\geq w : \hat{W}_t \geq h\}, \hat{W}_t = \sum_{i=t-w+1}^{t}\hat{L}_i,
    \label{detection statistic}
\end{equation}
\noindent where $\hat{W}_0=0$, $h$ represents the detection threshold, and $w$ refers to the cumulative length or window length. 

The window length $w$ is assumed to correspond to the tolerable delay of detection. This study assumes that the FMA is not operational during the first ($w-1$) observations. At a time $t$, we only need to compute $\hat{L}_t$. The LLR is equivalent to the following equation:

\begin{equation}
   \hat{L}_t = \log \frac{P_1(\boldsymbol{p}_t|\{\boldsymbol{p}_i\}_{i=1}^t)}{P_\infty(\boldsymbol{p}_t|\{\boldsymbol{p}_i\}_{i=1}^t)} = \big[\boldsymbol{p}_t-\boldsymbol{\mu}_t^p-\frac{1}{2}\boldsymbol{y}_t\big]^\top{\boldsymbol{\Sigma}_t^p}^{-1}\boldsymbol{y}_t.
   \label{LLR}
\end{equation}
\subsection{Statistical Performance}

This section theoretically investigates the statistical performance of the proposed stopping time $\tau$ defined by~\eqref{detection statistic}. Specifically, the worst-case probability of missed detection and the worst-case probability of a false alarm for any interval with a length $w_\alpha$ need to be examined. The exact calculation of these probabilities is a difficult problem, we turn to derive the performance bounds. These bounds are defined by the following Theorem.

\begin{theorem}
  Consider the FMA stopping time $\tau$. The worst-case probability of a false alarm over a given interval $w_\alpha$ is bounded by:
    \begin{equation}
        P_{FA} (\tau, w_\alpha) \leq \alpha (h, w_\alpha),
    \end{equation}
    where,
    \begin{equation}
        \alpha(h,w_\alpha) = 1-\bigg[P_\infty\big(\hat{W}_w <h\big)\bigg] ^{w_\alpha}.
        \label{false alarm}
    \end{equation}
    The upper bound of the worst-case probability of missed detection is given by:
    \begin{equation}
        P_{MD}(\tau, w) \leq \beta(h,w),
    \end{equation}
    where,
    \begin{equation}
        \beta(h,w)=P_1\big(\hat{W}_w<h\big).
        \label{missed detection}
    \end{equation}
    \label{FMAperformance}
\end{theorem}
\begin{proof}
    The proof is given in Section III of the Supplemental Material.
\end{proof}
Note that $\alpha$ and $\tilde{\alpha}$ in \eqref{false alarm constraint} are not equal, the value $\tilde{\alpha}$ is selected to satisfy $P_{FA}<\tilde{\alpha}$. Given the upper bound of $P_{FA}$ and $P_{MD}$, detection threshold $h$ can be determined to satisfy the false alarm constraint, as given in the following Corollary.
\begin{corollary}
    Assume $F_i, i=\{\infty, 1\}$ is a cumulative distribution function of detection statistic $\hat{W}_t$ when the SOI is absent ($F_\infty$) and present ($F_1$). Then, a possible detection threshold $h$ is given by:
    \begin{equation}
        h(\tilde{\alpha}) = F^{-1}_\infty\bigg[(1-\tilde{\alpha})^{\frac{1}{w_\alpha}}\bigg],
        \label{threshold}
    \end{equation}
    where $\tilde{\alpha}$ is the selected probability of a false alarm, $F_\infty^{-1}$ is the inverse function of $F_\infty$, and the upper bound of the $P_{MD}$ is given by:
    \begin{equation}
\beta(h(\tilde{\alpha}),w)=F_1\bigg[F_\infty^{-1}\big[(1-\tilde{\alpha})^\frac{1}{w_\alpha}\big]\bigg],
\label{beta}
    \end{equation}
    where $h$ is the predefined threshold determined based on~\eqref{threshold}.
\end{corollary}
\begin{proof}
\eqref{threshold} can be derived by direct application of~\eqref{false alarm} and~\eqref{missed detection}. Since $\tilde{h}<h(\tilde{\alpha})$ and $\beta(h,w)$ is a monotonically increasing function with respect to $h$, it holds that $P_{MD}(\tau,w_\alpha)\leq \beta(\tilde{h},w)\leq \beta(h, w)$. Then, we have~\eqref{beta}.
\end{proof}

Under the null hypothesis $(\mathcal{H}_0)$, $\boldsymbol{p}_t - \boldsymbol{\mu}_t^p$ is a multivariate Gaussian with a zero mean and a covariance matrix $\boldsymbol{\Sigma}_t^p$. The mean and variance of the LLR under $\mathcal{H}_0$ are denoted by $\mu_{\mathcal{H}_0}$ and $\sigma_{\mathcal{H}_0}^2$, respectively:

\begin{equation}
\begin{aligned}
        \mu_{\mathcal{H}_0} = -\frac{1}{2}\boldsymbol{y}_t{\boldsymbol{\Sigma}_t^p}^{-1}\boldsymbol{y}_t&,\sigma^2_{\mathcal{H}_0} = {\boldsymbol{y}_t}^\top {\boldsymbol{\Sigma}_t^p}^{-1} \boldsymbol{y}_t.\\
\end{aligned}
\label{PDF_H0}
\end{equation}

Under the alternative hypothesis ($\mathcal{H}_1$), the probability density function and its parameters can not be directly obtained. In the following lemma, the distribution of LLR under $\mathcal{H}_1$ is derived as follows.
\begin{lemma}
    Under $\mathcal{H}_1$, the LLR follows multi-variate Gaussian distribution with a mean ${\mu}_{\mathcal{H}_1,t}$ and $\sigma_{\mathcal{H}_1}^2$, where,
    \begin{equation}
    \begin{aligned}
    \mu_{\mathcal{H}_1,t} &= [\boldsymbol{e}_t-\frac{1}{2}\boldsymbol{y}_t]^\top{\boldsymbol{\Sigma}_t^p}^{-1}\boldsymbol{y}_t,\\
    \boldsymbol{e}_t& = \boldsymbol{y}_t - \sum_{m=1}^M \boldsymbol{C}^m \boldsymbol{\psi}_t^m,\\
    \boldsymbol{\psi}_t^m &= \boldsymbol{G}^m(\boldsymbol{\psi}_{t-1}^m+{\boldsymbol{K}^m}\boldsymbol{e}_{t-1}),\\
    \sigma_{\mathcal{H}_1}^2 &= \sigma_{\mathcal{H}_0}^2, 
    \end{aligned}
    \label{kalman post change}
    \end{equation}
  where $\boldsymbol{K}^m$ is the Kalman gain of the $m$th LDS, $\boldsymbol{e}_t$ represents the observation error at a time $t$, $\boldsymbol{y}_t$ refers to the deterministic SOI , and $\boldsymbol{\psi}_t^m$ is the predicted latent variable at a time $t$.
  \label{lemmapostchange}
\end{lemma}
\begin{proof}
    The derivation can be found in Section IV of the Supplemental Material.
\end{proof}

\begin{figure}
    \centering
    \includegraphics[width=3.5in]{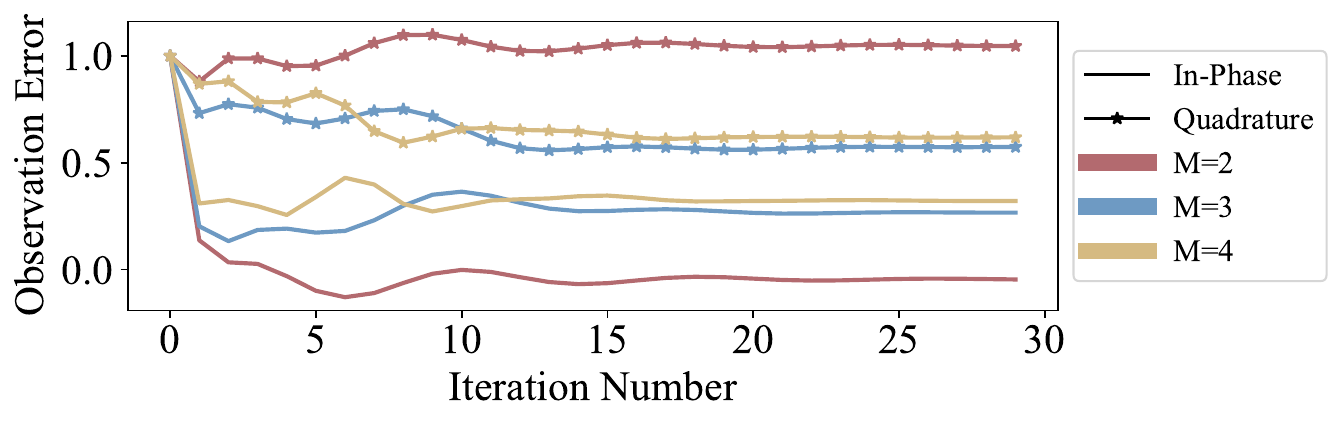}
    \caption{The observation error calculated by~\eqref{kalman post change}.}
    \label{mismatch convergence}
\end{figure}

In the Kalman filter, after sufficient iterations, the covariance matrix of the update variable $\hat{\boldsymbol{\Sigma}}_t$ converges to a certain value $\hat{\boldsymbol{\Sigma}}$. This convergence implies that the Kalman gain will also converge to a value $\boldsymbol{K}$. The covariance of the updated variable $\hat{\boldsymbol{\Sigma}}$ satisfies the algebraic Riccati equation~\cite{lancaster1995algebraic}. This study experimentally verifies that observation error $\boldsymbol{e}_t$ defined by~\eqref{kalman post change} converges to a certain value after a sufficient number of iterations. In the verification simulation tests, this study uses the same parameter settings as those described in Section V-A (see Page 9). As shown in Fig.~\ref{mismatch convergence}, the derived observation error $\boldsymbol{e}_t$ in the iteration function~\eqref{kalman post change} converges to a specific value. However, a theoretical investigation of the convergence of~\eqref{kalman post change} is beyond the scope of this paper and will not be presented in this work.

Given the above information, this study explicitly formulates the upper bound of the worst-case probability of missed detection and false alarm. Assuming $\tau$ is the stopping time under a threshold $h$; then, it holds that $P_{FA}(\tau,w_\alpha)<\tilde{\alpha}$. Next, let $\Phi(\cdot)$ be the cumulative function of the standard Gaussian distribution; then, it can be written that:
\begin{equation}
    \begin{aligned}
        h(\tilde{\alpha}) &= \sqrt{w\sigma_{\mathcal{H}_0}^2}\cdot\Phi^{-1}[(1-\tilde{\alpha})^{\frac{1}{w_\alpha}}]+w\mu_{\mathcal{H}_0},\\
        P_{MD}(\tau)&\leq \Phi\bigg(\frac{h-w\mu_{\mathcal{H}_1}}{\sqrt{w\sigma^2_{\mathcal{H}_1}}}\bigg),\\
        P_{FA}(\tau, w_\alpha) &\leq 1-\bigg[ \Phi\bigg(\frac{h-w\mu_{\mathcal{H}_1}}{\sqrt{w\sigma^2_{\mathcal{H}_1}}}\bigg)\bigg]^{w_\alpha},
    \end{aligned}
    \label{detection threshold}
\end{equation}
where $\mu_{\mathcal{H}_1}$ is the value of $\mu_{\mathcal{H}_{1,t}}$ when it converges. 

\section{Numerical Simulation Results}
This section first evaluates the effectiveness of the proposed IFLDS model and its parameter learning method. Then, without the loss of generality of the theoretical results, the Monte Carlo results for the FMA stopping time are compared with the theoretical result and baseline methods. The numerical simulation datasets and evaluation metrics are described in Section V-A. The simulation results and detailed discussions are presented in Sections V-B and V-C.
\subsection{Simulation Design}
\subsubsection{Data Description}
This study assumed there were $M$ background sources that transmit signals, each of which was represented by an LDS; the transition matrices $\boldsymbol{G}^m$ of the LDS and the corresponding output matrices for each background source were defined as follows:
\begin{equation}
    \boldsymbol{G}^m(\theta_m)=\epsilon_m \bigg[
    \begin{matrix}
       \cos(\theta_m) & -\sin(\theta_m)\\
        \sin(\theta_m) & \cos(\theta_m)
    \end{matrix}
\bigg],
\boldsymbol{C}^m = \bigg[
\begin{matrix}
    0.25 & -1.25\\
    -1 & -0.50
\end{matrix}
\bigg],
\end{equation}
\noindent where $\epsilon_m$ and $\theta_m$ are defined according to the assumption of four background sources as follows:
\begin{equation}
\begin{aligned}
    & \{(\epsilon_m,\theta_m)\}_{m=1}^4\\
 =& \{(0.95, 0),(0.9,\pi/2),(0.85, \pi/3),(0.75, \pi/6)\}.
\end{aligned}
\end{equation}
The covariance matrices $\boldsymbol{Q}$ and $\boldsymbol{R}$ are set as $0.01\boldsymbol{I}$.

The initial states of all four LDS were selected as $[0, 0]^\top$. In the RPL task, the CBS with a length of 2,000 was generated. In the TSD task, the settings of the CBS were the same as in the RPL task, except the SOI arrived at the 1,000th data point, lasted for 200 data points, and then disappeared, leaving only the CBS. The amplitude of the SOI $\boldsymbol{y}_t$ was a constant.

\subsubsection{Evaluation Metric}
The estimated number of background sources and the reconstruction errors were used in this study to evaluate the RPL task performance.

\textbf{Estimated Number $\hat{M}$}: The estimated number reflected the ability of the algorithm to estimate the number of background sources.

\textbf{Reconstruction Error $RE$}: The reconstruction error indicated the ability of the algorithm to estimate the model parameter $\{\boldsymbol{\Gamma}^m\}_{m=1}^M$, we define $ RE = \frac{1}{\bar{T}} \sum_{t=1}^{\bar{T}}|\boldsymbol{p}_t - \hat{\boldsymbol{p}}_t|^2$,
where $\boldsymbol{p}_t$ is the observed data sequence, and $\boldsymbol{\hat{p}}_t$ is the data sequence reconstruction using the estimated model parameters.

In the TSD task, the performance was evaluated using the probability of missed detection defined by~\eqref{inf_PMD} and the probability of false alarm defined by~\eqref{false alarm constraint}.

\subsubsection{Baseline Methods}
This study compared the performance of different approaches by designing the combination of the representation model for the RPL task and the stopping time for the TSD task.

In the RPL task, the comparison models were as follows:
\begin{enumerate}
    \item The MIBP was first proposed in~\cite{gael2008infinite}. We incorporate the MIBP prior to the proposed IFLDS to examine its ability to estimate the background source number;
    \item The HMM was used to represent the BS in~\cite{ford2021unknown}. The Baum-Welch algorithm~\cite{yang2017statistical} and the Viterbi algorithm \cite{viterbi1967error} were used to reconstruct the data;
    \item The LDS was used to represent the BS in~\cite{ford2019unknown}. The expectation-maximization (EM) algorithm~\cite{ghahramani1996parameter} and the Kalman filter \cite{khodarahmi2023review} were used to reconstruct the data.
\end{enumerate}

In the TSD task, the following four algorithms were used as the baseline methods in this work:
\begin{enumerate}
    \item The LDS-FMA used the LDS to represent the CBS, and employed the FMA stopping time; 
    \item The Gaussian-FMA used the Gaussian model to represent the CBS and employed the FMA stopping time; 
    \item The FLDS-CUSUM used the FLDS to represent the CBS and the CUSUM as a stopping time. The CUSUM stopping time~\cite{tartakovsky2014sequential} \cite{wald2004sequential} was defined as follows:
\begin{equation}
    \tau_{C}=\inf\{t\geq 1: \max_{1\leq k\leq t} \sum_{i=k}^t \hat{L}_i\geq h\};
    \label{CUSUM loss}
\end{equation}
    \item The FLDS-Shewhart used the FLDS to represent the CBS and adopted the Shewhart test as a stopping time. Shewhart test~\cite{moustakides2014multiple} was proven to be optimal for transient change detection when $w=1$; the Shewhart stopping time was defined by:
\begin{equation}
    \tau_{S} = \inf \{t\geq 1: \hat{L}_t \geq h\}.
\end{equation}
\end{enumerate}
\subsection{Functional Validation of RPL Task}
In this study, the RPL task was validated from two perspectives: the background source number estimation and the source parameter estimation. The evaluation metrics were computed on a per-dataset basis and averaged over $10^4$ Monte Carlo simulations. The number of background sources was gradually increased from two to four. In the simulations, the hyper-parameters was set to $\alpha = 1, \beta_0=2, \beta_1 = 0.1$, $\gamma_0=10, \gamma_1=1$, $\boldsymbol{M}_0=\boldsymbol{0}, \boldsymbol{K}_0 = \boldsymbol{I}, n_0=4, \boldsymbol{S}_0 = 0.75 \bar{\boldsymbol{S}}$, where $\bar{\boldsymbol{S}} = \frac{1}{M\bar{T}}\sum_{t=1}^{\bar{T}}(\boldsymbol{p}_t-\bar{\boldsymbol{p}})(\boldsymbol{p}_t-\bar{\boldsymbol{p}})^\top$, where $\bar{\boldsymbol{p}}$ is the average value of the observed CBS samples. The simulation results are shown in Fig.~\ref{RPL result}.
\begin{figure}
    \centering
    \includegraphics[width=3.5in]{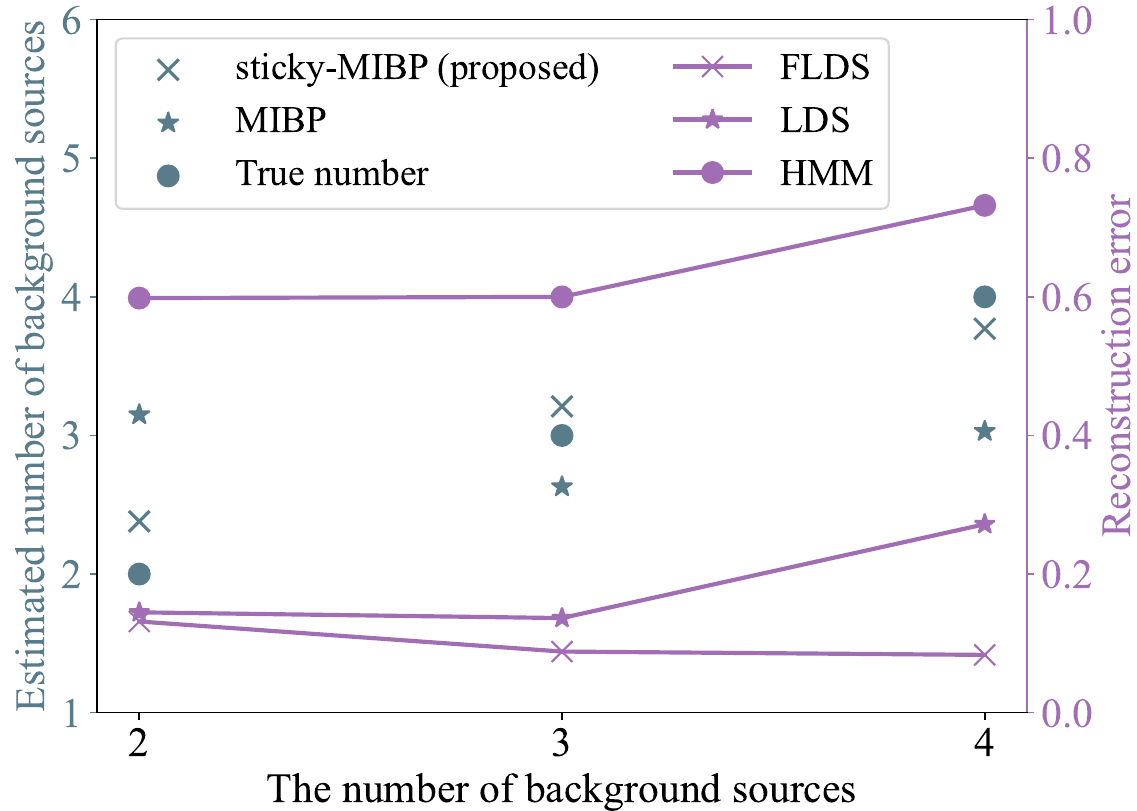}
    \caption{The estimation result of the background source number and the corresponding parameters.}
    \label{RPL result}
\end{figure}

In Fig.~\ref{RPL result}, the blue scatters represent the estimation result of the background source number. The results indicated that the proposed sticky-MIBP (blue forks) outperformed the conventional MIBP (blue stars) as the sticky-MIBP estimates were closer to the true background source number (blue points). The reason for this could be as follows. The conventional MIBP assumes each Markov chain $\boldsymbol{S}^m$ has a similar transition probability to other states, resulting in the frequently active-idle state transition. This circumstance is not suitable for the scenario when each source consistently exists.

The purple dotted line in Fig.~\ref{RPL result} illustrates the result of the data reconstruction of the estimated model parameter. As illustrated in Fig.~\ref{RPL result}, the HMM (five hidden states) demonstrated the weakest data representation capability among all methods, achieving the highest MSE value. When the number of background sources was two or three, the LDS and FLDS exhibited comparable performance, with an MSE of approximately 0.1. However, when the number of background sources increased to four, the performance of the LDS model decreased, resulting in a larger MSE value. The large reconstruction error in the LDS was primarily due to the limitation of the representation capability. 

In conclusion, by incorporating the sticky control, the proposed IFLDS can achieved more accurate number estimation results. The reconstruction error results validated the effectiveness and superiority of the proposed algorithm, especially in scenarios with multiple background sources.

\subsection{TSD Task Evaluation}
To verify the effectiveness of the TSD task, this study first calculated the probability of missed detection by varying the Signal to Interference and Noise Ratio (SINR) value. Then, the impact of a varying window length $w$ of the FMA test was examined. All simulations were based on the $10^4$ Monte Carlo trials. 
\subsubsection{Missed Probability Minimization Criterion Evaluation}
In this analysis, the SOI $\boldsymbol{y}_t$ consisted of 20 values equally spaced between 0.001 and 1.5, resulting in various SINRs. The probability of false alarms was $P_{FA} = 10^{-2}$ when $w_\alpha = 1000$, and the window length $w$ was set to 200. The detection thresholds of the baseline methods were determined by sweeping through different values to achieve the desired probability of false alarms. The detection results obtained for different background source numbers are shown in Fig.~\ref{Detection Performance}, where the \textit{x}-axis represents the SINR.
\begin{figure*}[htb]
	
	\begin{minipage}{0.33\linewidth}
		\hspace*{-3pt}
		\centerline{\includegraphics[width=\textwidth]{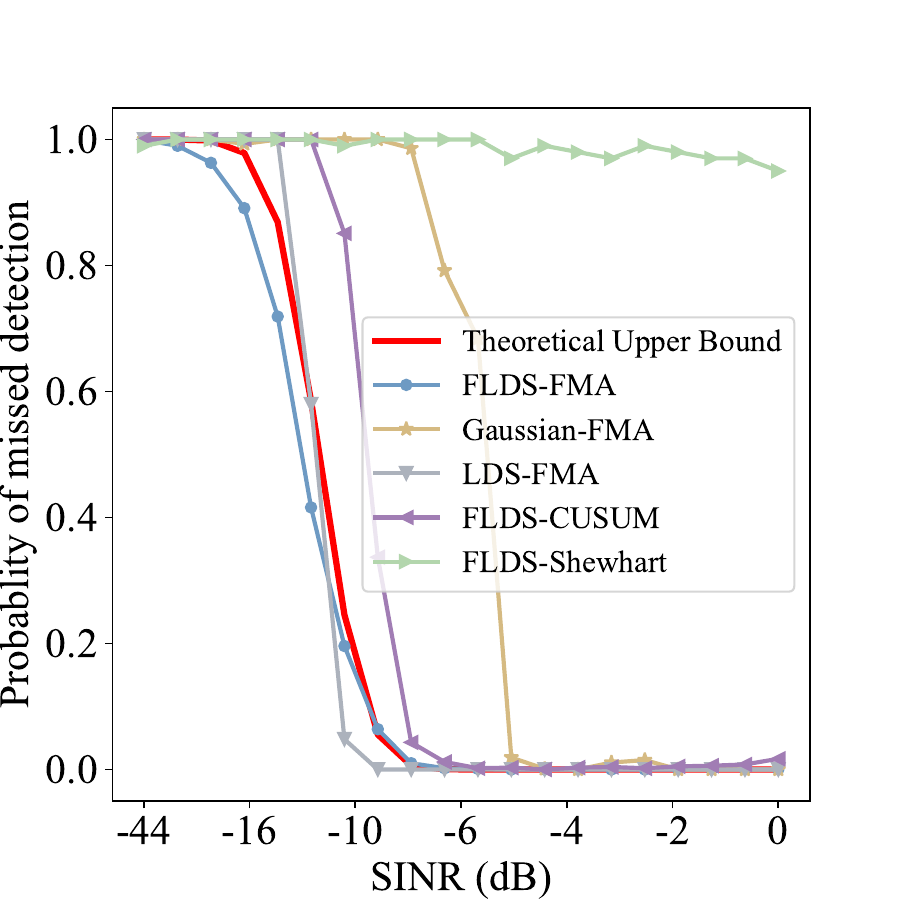}}
		\centerline{(a)}
	\end{minipage}
	\begin{minipage}{0.33\linewidth}
		\hspace*{-3pt}
		\centerline{\includegraphics[width=\textwidth]{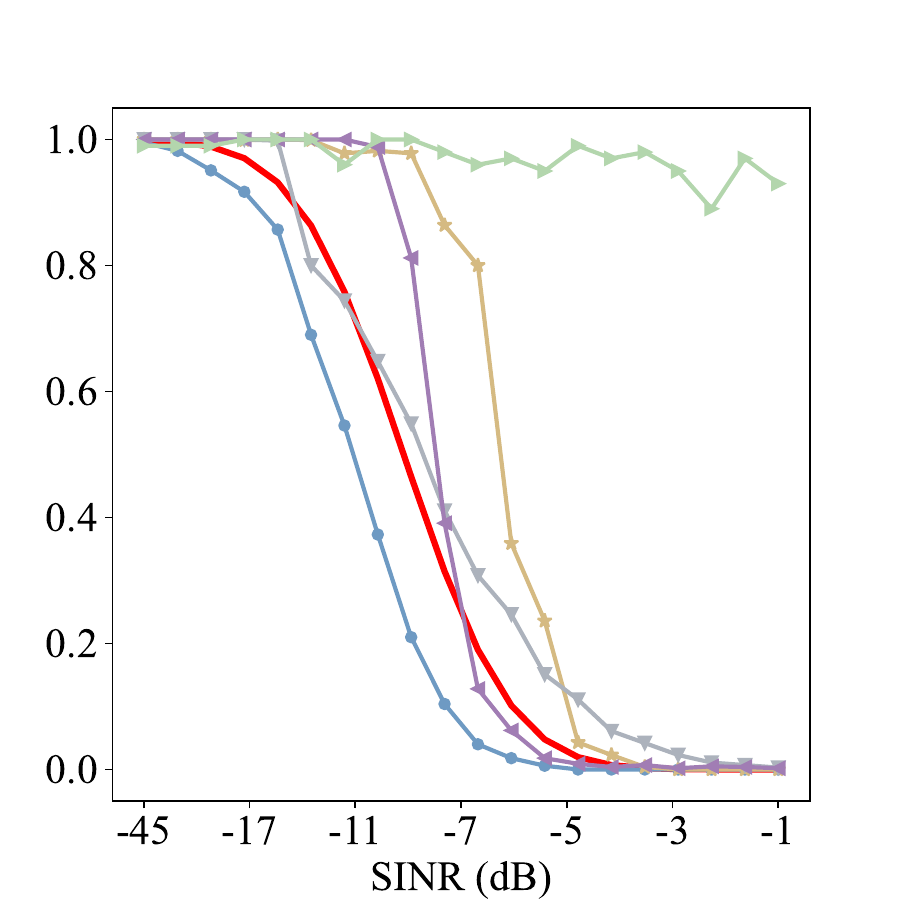}}
	 
		\centerline{(b)}
	\end{minipage}
	\begin{minipage}{0.33\linewidth}
		\hspace*{-3pt}
		\centerline{\includegraphics[width=\textwidth]{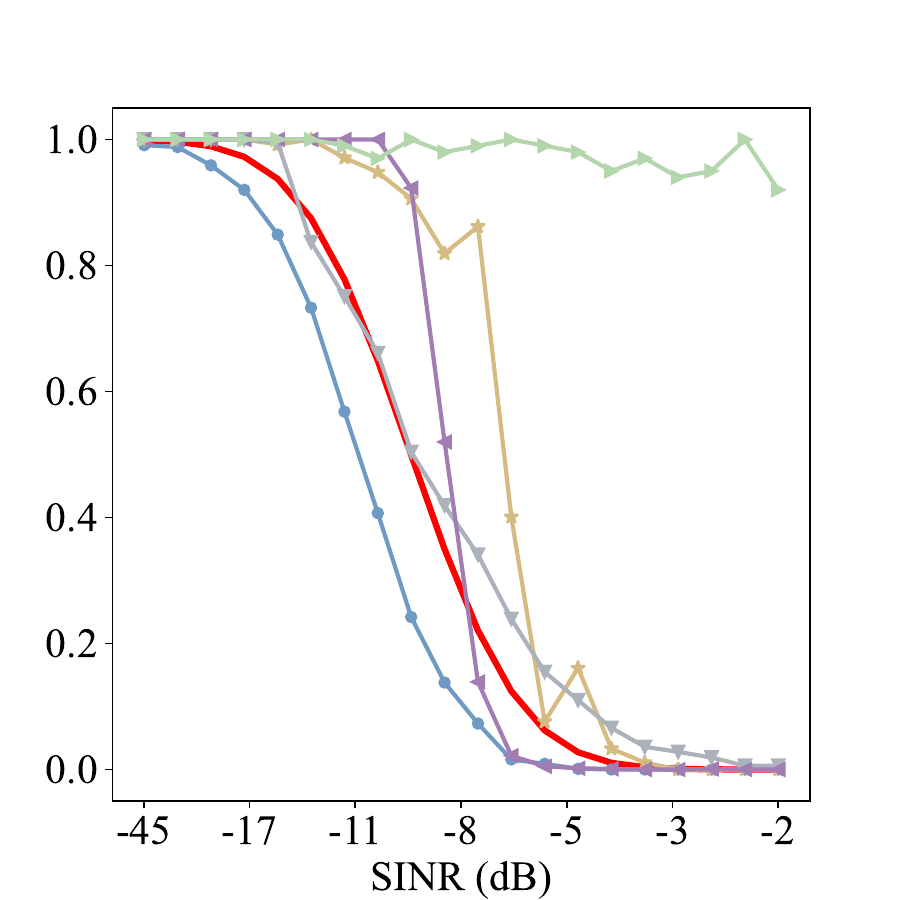}}
	 
		\centerline{(c)}
	\end{minipage}
	 
	\caption{Detection performance for different background source numbers. The red line in each figure represents the theoretical upper bound of the proposed FLDS-FMA test defined by~\eqref{detection threshold}; the \textit{x}-axis indicates the SINR, which was calculated empirically: (a) $M=2$; (b) $M=3$; (c) $M=4$.}
	\label{Detection Performance}
\end{figure*}

As shown in Fig.~\ref{Detection Performance}, the Shewhart test exhibited the weakest performance among all methods due to its strict assumptions. Particularly, the Shewhart test assumed a signal duration of one. The CUSUM-based algorithm accumulated all previous data to make decisions, which made it not sensitive to the SOI, yielding degraded performance. The Gaussian-FMA assumed that the observed data were i.i.d. and followed a Gaussian distribution, which contradicted the time-correlated characteristics. Consequently, the Gaussian-FMA achieved the second weakest performance, with non-monotonic fluctuations in its performance curve, as shown in Fig.~\ref{Detection Performance}~(c). In contrast, the FLDS-FMA demonstrated superior performance, achieving the lowest probability of missed detection across all SNR values. 

Comparing the results in Figs.~\ref{Detection Performance}~(a)--\ref{Detection Performance}~(c), as the number of background sources increased, the performance gap between the FLDS-FMA method and the other methods increased with the number of background sources. Overall, the proposed method demonstrated exceptional performance, especially in scenarios with multiple background sources. 

\subsubsection{Window Length Influence Analysis}
Setting a window length $w$ is a challenging problem in practice. In this study, the influence of the window length value was evaluated, and the effectiveness of the proposed statistic was validated. Particularly, the window length $w$ was set to 10, 100, and 200, and $w_\alpha=1000$. The number of background sources was four, and the SOI amplitude was 0.48, which corresponded to the SINR value of -12.03~dB. The baseline methods were two FMA-based methods, the Gaussian-FMA method and the LDS-FMA method. Fig.~\ref{PMDvsPFAinWL} presents the $P_{MD}$ versus $P_{FA}$ curves of the proposed FLDS-FMA and the baseline methods. For the window length of 10, all methods showed comparable performance, with the FLDS-FMA slightly outperforming the other methods. As the window length increased, this small advantage of the FLDS-FMA accumulated, and its probability of missed detection diverged further from those of the other two algorithms, achieving the best performance. Meanwhile, the performance of the Gaussian-FMA algorithm remained stable, with a missed detection probability of approximately 90\% for $P_{FA}=0.1$. In addition, the non-monotonic behavior in the Gaussian-FMA curve indicated that the Gaussian model could not represent the CBS accurately.

Further, in Figs.~\ref{PMDvsPFAinWL}~(a)-\ref{PMDvsPFAinWL}~(c), there were gaps between the FLDS-FMA performance curve (blue dotted line) and the theoretical upper bound (red solid line). These gaps were small when the probability of a false alarm approached zero; however, when the probability of a false alarm increased, these gaps widened. To investigate the reasons behind this gap, this study plotted the theoretical and empirical density curves and detection statistics under both hypotheses, as shown in Fig.~\ref{StatandDensity}.

\begin{figure*}
    \centering
    \includegraphics[width=\linewidth]{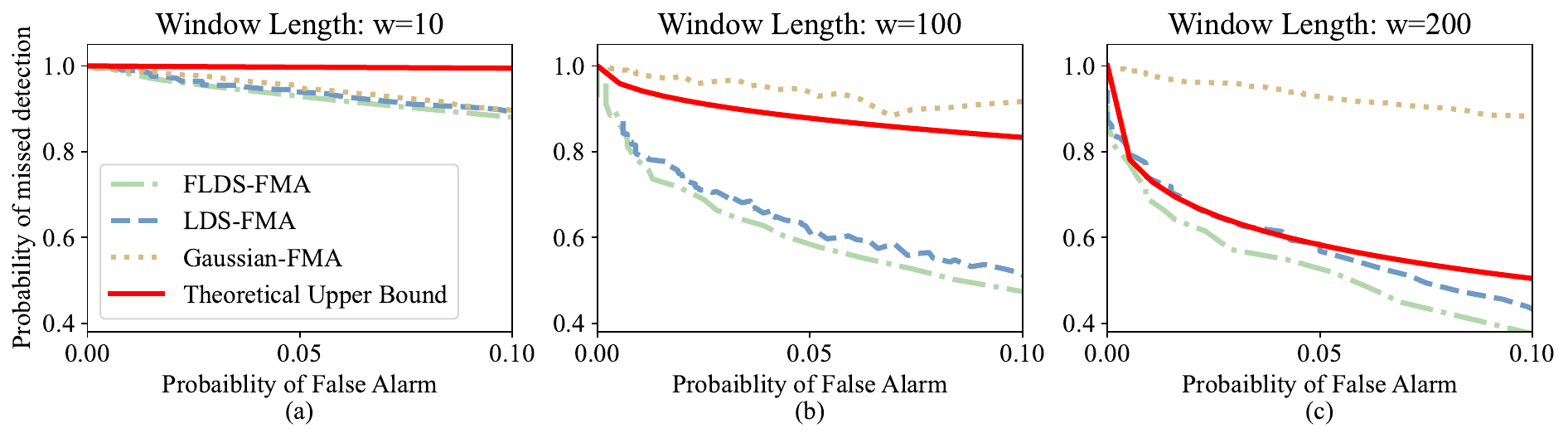}
    \caption{The $P_{MD}$ versus $P_{FA}$ curves of the FLDS-FMA method and baseline methods for various values of window length $w$. There were four background sources, and the SOI amplitude was $0.48$, which corresponded to -12.03~dB; (a): $w=10$, (b): $w=100$, (c): $w=200$.}
    \label{PMDvsPFAinWL}
\end{figure*}
Fig.~\ref{StatandDensity} summarizes the results of the proposed FLDS-FMA detection method defined by~\eqref{detection statistic} for a window length of $w=10,100,200$. The detection statistic (purple line) is presented in Fig.~\ref{StatandDensity}, where orange areas indicate signal absence and blue areas refer to the signal presence. The empirical distribution of the detection statistic obtained by the kernel density estimation (KDE) was compared with the theoretical distribution defined by~\eqref{PDF_H0}~and~\eqref{kalman post change}. The empirical distributions obtained by the KDE are indicated as solid lines in Fig.~\ref{StatandDensity}, and the theoretical distributions are represented by dashed lines.

The results in Fig.~\ref{StatandDensity} demonstrated that the empirical density (orange solid line) aligned well with the theoretical density (orange dashed line) under the null hypothesis. However, under the alternative hypothesis, the empirical density (blue solid line) did not match the theoretical counterpart (blue dashed line) as closely, showing a slight distortion of the density curve. This phenomenon could be observed for $w=100,200$. This distortion was primarily due to the cumulative phase in the detection statistic, where the statistic gradually increased when the signal was present (the blue region in Fig.~\ref{StatandDensity}). As the probability of a false alarm approached zero, the probability of missed detection tended to be one. However, as the probability of a false alarm began to increase, the reduction in the probability of missed detection was less than the ideal (theoretical) one, yielding a performance gap. In fact, one of the relaxations made when deriving the performance upper bound ignored the accumulation stage of the detection statistic, as shown in~(III.11) in the Supplemental Material. 

Thus, it could be concluded that as the window length increased, the distribution curves corresponding to the signal absent and its presence diverged, leading to a lower probability of missed detection. In addition,~\eqref{detection threshold} provided a viable approach for performance prediction, with signal duration incorporated into the performance calculation process.

\begin{figure*}
    \centering
    \includegraphics[width=\linewidth]{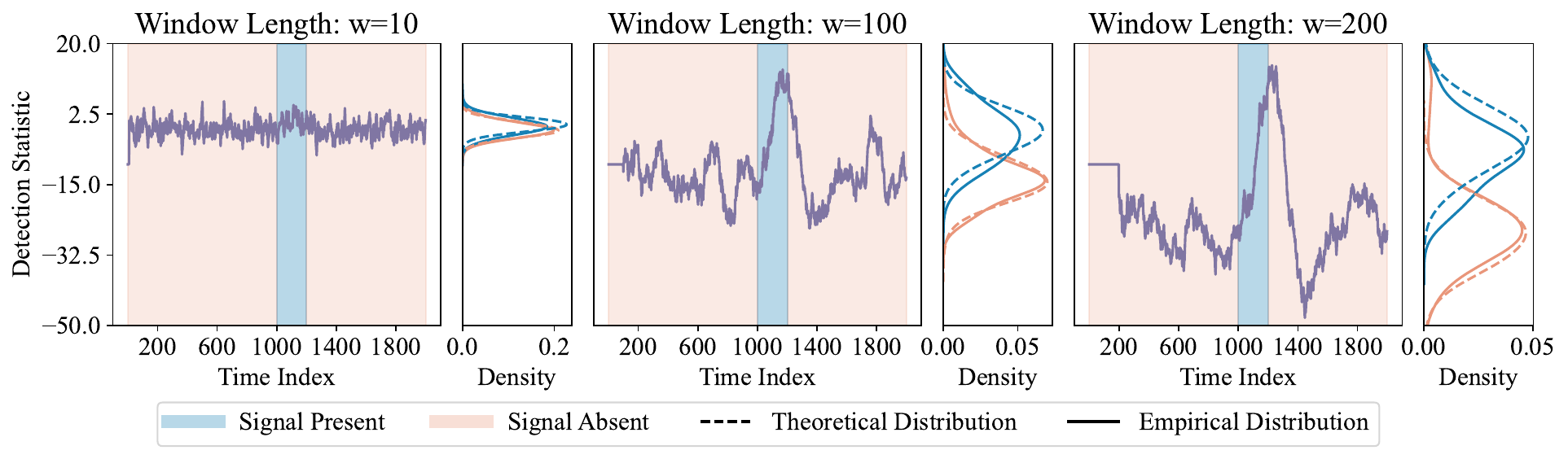}
    \caption{The distributions of the FLDS-FMA detection statistics for cases when the signal was absent and present. The theoretical distributions under the null hypothesis and alternative hypothesis are obtained by~\eqref{PDF_H0} (orange curve) and~\eqref{kalman post change} (blue curve).}
    \label{StatandDensity}
\end{figure*}
\section{Pulse Signal Detection Effect under Communication Interference}
In this section, the proposed FLDS-FMA method was compared with the LDS-FMA and Gaussia-FMA in the context of radar signal detection in the presence of communication interference. 

The CBS was the communication signal consisting of two communication sources: one using the binary phase shift keying (BPSK) modulation and the other employing the Quadrature Phase Shift Keying (QPSK)  modulation. SOI was emitted by radar. The pulse repetition frequency (PRF) was 3 kHz, the duty cycle was 8\%, and the pulse width (PW) of the SOI was 20~$\mu s$ (i.e., 200 samples). The ratio of the power of the SOI to the power of the communication interference signal and the environmental noise across the entire 5 MHz bandwidth was -10.02 dB, which is the SINR value. The Gaussian-FMA and LDS-FMA methods were used for performance comparison.

The parameters of the IFLDS, LDS, and Gaussian models were trained based on 4,000 signal samples. Specifically, the IFLDS was trained with 300 iterations using the proposed PGAS algorithm; the LDS was trained with 300 iterations using the EM algorithm, and the parameters of the Gaussian distribution were estimated using the maximum likelihood estimation. The detection performance was evaluated on 51,200 signal samples, as depicted in Fig.~\ref{COM}. The detection window length $w$ was set to 200, and the detection threshold $h$ was adjusted to maintain a false alarm probability of 15\%, which is indicated by the red line in the second to fourth panels of Fig.~\ref{COM}.

In the first and second panels of Fig.~\ref{COM}, the blue region indicates the area where the signal was present. The third to fifth panels of Fig.~\ref{COM} illustrate the detection statistics of the FLDS, LDS, and Gaussian-based detectors, respectively. The results indicated that the proposed FLDS-FMA algorithm outperformed than the baseline methods. In particular, the $P_D$ of the FLDS, LDS, and Gaussian-based detectors were 63\%, 41\%, and 33\% when $P_{FA}=15\%$, respectively. 

\begin{figure*}
    \centering
    \includegraphics[width=\linewidth]{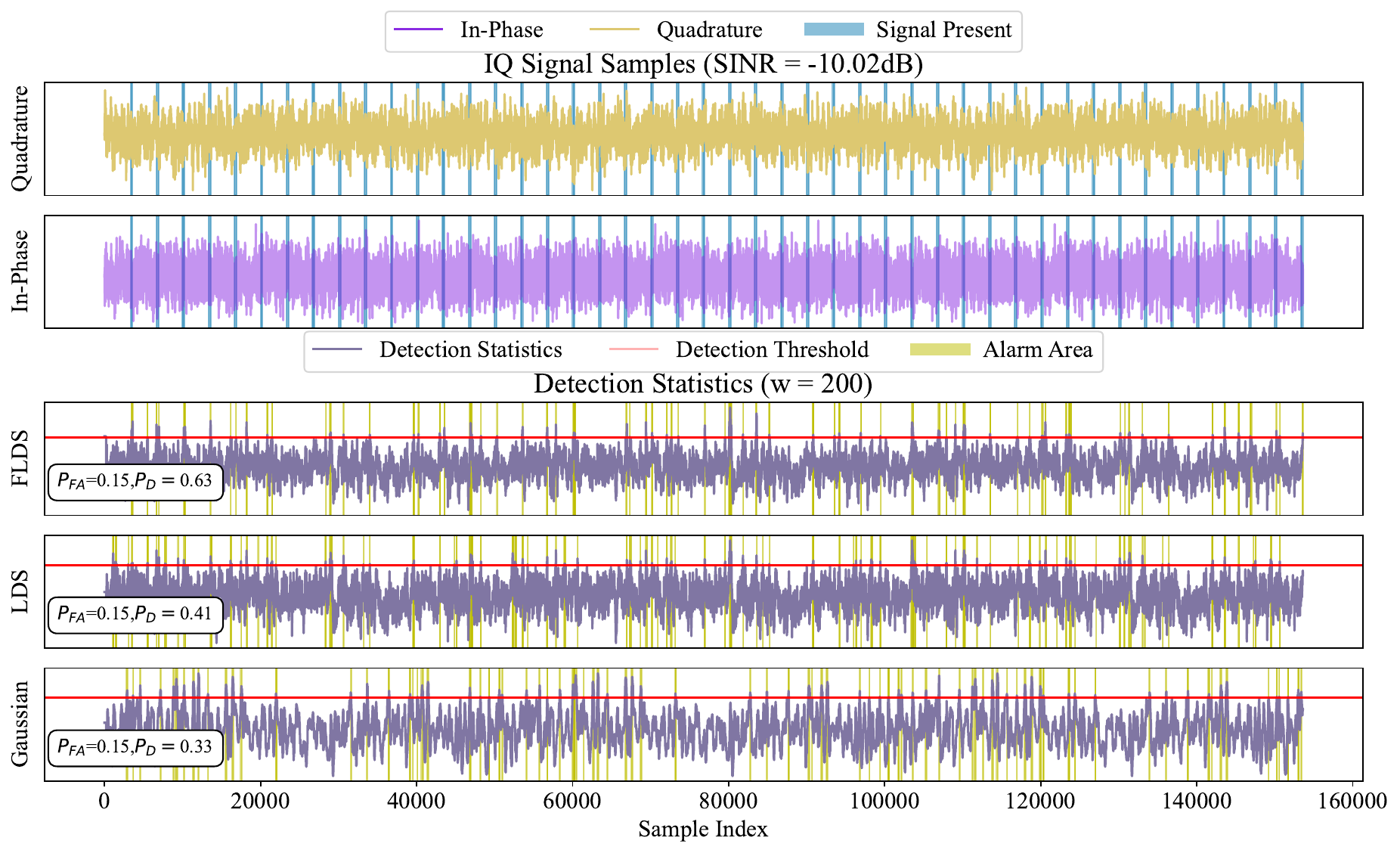}
    \caption{The IQ data (the first and second panels) and the detection statistic (the third to fifth panels); the third panel: FLDS-FMA statistic; the fourth panel: LDS-FMA statistic; the fifth panel: Gaussian-FMA statistic.}
    \label{COM}
\end{figure*}

\section{Conclusion}

This paper proposes a novel approach to detect transient signals under the background of IFLDS. First, the IFLDS is introduced to represent CBS. In addition, a sticky IMBP prior is designed to incorporate sticky control, which improves the model’s representation capability. Then, the IFLDS parameters are determined using the SS and PGAS, allowing for automatic estimation of the number of background sources. Further, an FMA stopping time is designed to detect transient signals, and its statistical performance is analyzed. To address the computational challenges in implementing the FMA stopping time, this study derives a closed-form recursive function for likelihood computation. Finally, a dependence structure between the underlying model is developed to compute the LLR feasibly. Numerical simulations and experiments verified the effectiveness and the superiority of the proposed method.

\label{sec:conclusion}

\bibliographystyle{IEEEtran}
\bibliography{ref}




\end{document}